\documentclass[journal]{IEEEtran}
\usepackage{amsmath,amssymb,amsfonts,amsthm}
\usepackage{algorithmic}
\usepackage{algorithm}
\usepackage{array}
\usepackage[caption=false,font=normalsize,labelfont=sf,textfont=sf]{subfig}
\usepackage{textcomp}
\usepackage{stfloats}
\usepackage{url}
\usepackage{booktabs}
\usepackage{tabularx}
\usepackage{multirow}
\usepackage{verbatim}
\usepackage{graphicx}
\usepackage{adjustbox}
\usepackage{cite}
\usepackage{acronym}
\usepackage{orcidlink}
\usepackage{csquotes}
\usepackage{subfig}
\newcommand{\BE}{\begin{equation}\begin{aligned}}
\newcommand{\EE}{\end{aligned}\end{equation}}
\newcommand{\BEs}{\begin{equation*}\begin{aligned}}
\newcommand{\EEs}{\end{aligned}\end{equation*}}
\usepackage{tkz-euclide}
\usepackage{tikz,pgfplots}
\usepgfplotslibrary{groupplots,units}
\usetikzlibrary{patterns}
\pgfplotsset{compat=newest}
\usetikzlibrary{shapes.geometric, arrows, calc, fit, shapes, backgrounds}

\acrodef{asd}[ASD]{anomalous sound detection}
\acrodef{gmm}[GMM]{Gaussian mixture model}
\acrodef{dg}[DG]{domain generalization}
\acrodef{auc}[AUC]{area under the \ac{roc} curve}
\acrodef{pca}[PCA]{principal component analysis}
\acrodef{lda}[LDA]{linear discriminant analysis}
\acrodef{lime}[LIME]{local interpretable model-agnostic explanations}
\acrodef{slime}[SLIME]{sound \ac{lime}}
\acrodef{dnn}[DNN]{deep neural network}
\acrodef{cd}[CD]{cosine distance}
\acrodef{ln}[LN]{length normalization}
\acrodef{rise}[RISE]{randomized input sampling for explanation}
\acrodef{roc}[ROC]{receiver operating characteristic}
\acrodef{tsne}[t-SNE]{t-distributed stochastic neighbor embedding}
\acrodef{xai}[xAI]{explainable artificial intelligence}
\acrodef{pauc}[pAUC]{partial area under the \ac{roc} curve}
\acrodef{umap}[UMAP]{uniform manifold approximation and projection}
\acrodef{oe}[OE]{outlier exposure}
\acrodef{im}[IM]{inlier modeling}
\acrodef{ic}[IC]{intra-class}
\acrodef{cce}[CCE]{categorical cross-entropy}

\theoremstyle{plain}
\newtheorem{thm}{Theorem}
\newtheorem{lem}[thm]{Lemma}

\newtheorem{cor}[thm]{Corollary}
\theoremstyle{definition}
\newtheorem{defn}[thm]{Definition}

\theoremstyle{remark}
\newtheorem*{rem}{Remark}

\DeclareMathOperator*{\argmax}{arg\,max}

\DeclareMathOperator*{\softmax}{smax}

\begin{document}

\title{Why do Angular Margin Losses work well for Semi-Supervised Anomalous Sound Detection?}

\author{Kevin Wilkinghoff \orcidlink{0000-0003-4200-9129},~\IEEEmembership{Student Member,~IEEE}, and Frank Kurth  \orcidlink{0000-0002-9992-083X},~\IEEEmembership{Senior Member,~IEEE}
\thanks{The authors are with Fraunhofer FKIE, Fraunhoferstraße 20, 53343 Wachtberg, Germany  (e-mail: kevin.wilkinghoff@ieee.org, frank.kurth@fkie.fraunhofer.de).}
}



\maketitle

\begin{abstract}
State-of-the-art anomalous sound detection systems often utilize angular margin losses to learn suitable representations of acoustic data using an auxiliary task, which usually is a supervised or self-supervised classification task.
The underlying idea is that, in order to solve this auxiliary task, specific information about normal data needs to be captured in the learned representations and that this information is also sufficient to differentiate between normal and anomalous samples.
Especially in noisy conditions, discriminative models based on angular margin losses tend to significantly outperform systems based on generative or one-class models.
The goal of this work is to investigate why using angular margin losses with auxiliary tasks works well for detecting anomalous sounds.
To this end, it is shown, both theoretically and experimentally, that minimizing angular margin losses also minimizes compactness loss while inherently preventing learning trivial solutions.
Furthermore, multiple experiments are conducted to show that using a related classification task as an auxiliary task teaches the model to learn representations suitable for detecting anomalous sounds in noisy conditions.
Among these experiments are performance evaluations, visualizing the embedding space with t-SNE and visualizing the input representations with respect to the anomaly score using randomized input sampling for explanation.
\end{abstract}


\section{Introduction}
\IEEEPARstart{S}{emi-supervised} \ac{asd} is the task of reliably detecting anomalous sounds while only having access to normal sounds for training a model \cite{aggarwal2017outlier}.
Since anomalies occur only rarely by definition and usually are very diverse, collecting realistic anomalous samples for training a system is much more difficult and thus more costly than collecting normal data.
Hence, a semi-supervised \ac{asd} setting is more realistic than a supervised \ac{asd} setting, for which anomalous sounds are available for training, because it substantially simplifies the data collection process.
There are also unsupervised \ac{asd} settings, for which the training dataset may also contain anomalous samples and it is unknown whether a training sample is normal or anomalous.
But for many applications, it can be ensured that only normal samples are collected for training and thus a semi-supervised setting can be assumed.
\par
\ac{asd} has many applications.
Examples are machine condition monitoring \cite{koizumi2020description,kawaguchi2021description,dohi2022description}, medical diagnosis \cite{murthy2021deep,dissanayake2021robust}, bioacoustic monitoring {\cite{ntalampiras2021acoustic,cejrowski2021buzz}, intrusion detection in smart home environments \cite{zieger2009acoustic} and detecting crimes \cite{foggia2016audio,li2018anomalous} or accidents \cite{valenzise2007scream,hayashi2018anomalous}.
Furthermore, detecting anomalous samples can also be understood as a subtask in acoustic open-set classification \cite{shon2019mce,mesaros2019acoustic,naranjo2022open}.
Throughout this work, we will use machine condition monitoring in domain-shifted conditions as an application example \cite{dohi2022description}.
Here, the audio signals may contain one or several of the following three components: 1) normal machine sounds, 2) anomalous machine sounds and 3) background noise consisting of a mixture of many other sound events.
The major difficulty of this \ac{asd} application is that anomalous components of machine sounds can be very subtle when being compared to the background noise making it difficult to reliably detect anomalous signal components.
Furthermore, machine sounds and background noise can change substantially for different domain shifts, which we define as alterations in the (acoustic) environment or changes in parameter settings of the machines.
The \ac{asd} system still needs to only detect anomalous signal components without frequently raising false alarms caused by any domain shift.
\par
There are several strategies to train an \ac{asd} system for machine condition monitoring using only normal data.
Among these strategies are generative models such as autoencoders \cite{marchi2017deep,koizumi2019unsupervised,suefusa2020anomalous,giri2020group,kapka2020id,wichern2021anomalous} or normalizing flows \cite{dohi2021flow,dohi2022disentangling} that directly try to model the probablity distribution of normal data, which is also called \ac{im} \cite{kawaguchi2021description}.
Another strategy is to use an auxiliary task, usually a classification task, for training a model to learn meaningful representations of the data (embeddings) that can be used to identify anomalies.
Possible auxiliary tasks for machine condition monitoring are classifying between machine types \cite{giri2020self,lopez2020speaker,inoue2020detection,zhou2020,wilkinghoff2021sub} or, additionally, between different machine states and noise settings \cite{wilkinghoff2021combining,venkatesh2022improved,nishida2022anomalous,wilkinghoff2023design}, recognizing augmented and not augmented versions of normal data (self-supervised learning) \cite{giri2020self} or predicting the activity of machines \cite{nishida2022anomalous}.
Using an auxiliary task to learn embeddings is also called \ac{oe} \cite{hendrycks2019deep} because normal samples belonging to other classes than a target class can be considered as proxy outliers \cite{primus2020anomalous}.
Often an angular margin loss such as SphereFace \cite{liu2017sphereface}, CosFace \cite{wang2018cosface} or ArcFace \cite{deng2019arcface} is utilized for training an \ac{oe} model.
Systems based on embeddings pre-trained on very large datasets \cite{grollmisch2020iaeo3,wilkinghoff2020anomalous,mueller2021acoustic} can be used, too.
However, it has been shown that directly training a system on the data yields better \ac{asd} results, even when only very limited training data is available \cite{wilkinghoff2023using}.
In addition, different strategies can be combined by using an ensemble of multiple models \cite{lopez2021ensemble,kuroyanagi2021ensemble,dang2022ensemble}.
\par
Different strategies to train an \ac{asd} system have different strengths and weaknesses.
Using an auxiliary task for training relies on additional meta-information to generate labels for a classification task whereas \ac{im}-based models do not need any labels.
Furthermore, autoencoders can localize anomalies in the input space by visualizing an element-wise reconstruction error as done in \cite{kapka2020id,suefusa2020anomalous}.
However, training ASD models by using an auxiliary task usually enhances their performance \cite{fernandez2021using}.
Even for \ac{im}-based models, performance can be significantly improved when utilizing meta information such as machine types.
In \cite{kapka2020id} a class-conditioned autoencoder is used, in \cite{kuroyanagi2021ensemble} not only spectral features but also the machine ID  is encoded and decoded, and in \cite{dohi2021flow} a normalizing flow is trained to assign lower likelihood to sounds of other machines and a higher likelihood to sounds of the target machine.
 As suspected in \cite{nishida2022anomalous,wilkinghoff2023design}, the most likely reason for the difference in performance is that, as stated before, recordings for machine condition monitoring are very noisy because of factory background noise.
This is a problem for \ac{im}-based models because they cannot tell the difference between arbitrary sound events not emitted by a monitored machine and normal or anomalous sounds emitted by the machine.
Both are considered equally important by the model.
Moreover, anomalies present in these noisy audio recordings are usually very subtle when being compared to the noise or other sound events present in a recording making it even more difficult to detect potential anomalies.
When being trained with an auxiliary task, a model learns to ignore noise, which can be assumed to be similar for all considered classes, and therefore to isolate the target machine sound by ignoring the uninformative background sound events.
As a result, these models are more sensitive to changes of the machine sounds and have better anomaly detection capabilities.
\par
Localizing and visualizing frequencies or temporal regions of recordings that are being considered anomalous is important for practical applications because users can better understand the decisions of the \ac{asd} system (\ac{xai} \cite{holzinger2020xai}).
Furthermore, this may help to find the cause of mechanical failure and thus can simplify the maintenance process.
As stated before, autoencoders can easily localize anomalies by using an element-wise reconstruction error.
Additional investigations on visualizing and explaining \ac{asd} decisions include showing that decisions of \ac{asd} systems for machine condition monitoring largely rely on high-frequency information \cite{mai2022explaining}.
This has been visualized using \ac{lime} \cite{ribeiro2016why} applied to sounds (\acs{slime})\cite{mishra2017local}.
Furthermore, \ac{umap} \cite{mcinnes2018umap} has been used to visualize representations of the data such as stacked consecutive frames of log magnitude spectrograms, log-mel magnitude spectrograms, or openL3 embeddings \cite{fernandez2021using}.
\par
The goal of this work is to explain why angular margin losses work well for anomalous sound detection.
To achieve this goal, the following contributions are made:
First and foremost, it is theoretically proven that, after normalizing the embedding space, training an \ac{asd} model by minimizing an angular margin loss using an auxiliary task can be considered as minimizing a regularized one-class loss while being less affected by noise or non-target sound events present in the data.
Moreover, it is experimentally verified that using an angular margin loss for training a model to discriminate between classes of an auxiliary task also leads to better \ac{asd} performance and thus is a better choice for an \ac{asd} task than minimizing a one-class loss such as an \ac{ic} compactness loss with a single or multiple classes.
Last but not least, a procedure for visualizing normal and anomalous regions of the input representations based on \ac{rise} is presented.
Using these visualizations, it is shown that normal and anomalous sounds cannot be distinguished from the highly complex background noise when training with a one-class loss.
In contrast, when using an auxiliary task with multiple classes the model learns to ignore noise and isolate the targeted machine sound for monitoring their condition.
\par
The paper is structured as follows:
In Section \ref{sec:losses}, various one-class losses and angular margin losses are reviewed.
Section \ref{sec:theory} presents our main theoretical results about the relation between these loss functions.
Section \ref{sec:exp} contains a description of the experimental setup and all experimental evaluations consisting of performance evaluations, a comparison between losses during training, visualizing normal and anomalous regions of input representations as perceived by the system and visualizing the resulting embedding spaces.
Section \ref{sec:conclusions} consists of the conclusions of this work.

\section{Loss Functions}
\label{sec:losses}
In this section, a unified presentation and discussion of several loss functions that are needed for presenting one of the main results of this work in Sec. \ref{sec:theory} will be given.
The following notation will be used throughout the paper:
$X$ denotes the space of input data samples, $N\in\mathbb{N}$ the number of classes defined for an auxiliary task and $D\in\mathbb{N}$ the dimension of the embedding space.

\subsection{One-Class Classification Losses}
\label{sec:comp}
When training a model for \ac{asd} while only having access to normal data i.e. a single class, this is referred to as \emph{one-class classification} and is some form of \ac{im}.
The compactness loss \cite{ruff2018deep}, whose goal it is to project the data into a hypersphere of minimum volume, will serve as a representative of losses for one-class classification and is defined as follows.
\begin{defn}[Compactness loss]
\label{def:comp}
Let $Y\subset X$ be finite.
Let $\mathcal{P}$ denote the power set, $\Phi$ denote the space of network architectures for extracting embeddings and $W(\phi)$ denote the parameter space of $\phi\in\Phi$, i.e. ${\phi:X\times W(\phi)\rightarrow\mathbb{R}^D}$.
Then, the \emph{compactness loss} is defined as
\BE&\mathcal{L}_\text{comp}:\mathcal{P}(X)\times \mathbb{R}^D \times\Phi\times W\rightarrow\mathbb{R}_+\\
&\mathcal{L}_\text{comp}(Y,c,\phi,w):=\frac{1}{\lvert Y\rvert}\sum_{x\in Y}\lVert \phi(x,w)-c\rVert_2^2.\EE
The vector $c\in\mathbb{R}^D$ is called \emph{center}.
\end{defn}
After training,  the (squared) Euclidean distance between the embedding of a given sample and the center can be utilized as an anomaly score: A greater distance indicates a higher likelihood for the sample to be anomalous.
A \emph{trivial solution} for minimizing the compactness loss with center $c\in\mathbb{R}^D$ is a parameter setting $w_c\in W(\phi)$ such that $\phi$ is the constant function ${\phi(x,w_c)=c}$ for all ${x\in X}$.
It is of utmost importance to prevent that the model to be trained is able to learn such a trivial solution.
Otherwise it is impossible to differentiate between normal and anomalous samples.
\par
There are several strategies to prevent a model from learning a trivial solution.
First of all, it needs to be ensured that ${c\neq c_0\in\mathbb{R}^D}$ where $c_0=\phi(x,w_0)$ is defined as the output of the network obtained by setting the weight parameters of model $\phi$ to zero. This is because we have $\phi(x,w_0)=c_0$ for all $x\in X$ as long as the model uses only linear operators, e.g. dense or convolutional layers, and all activation functions have zero as a fixed point, which is the case for most commonly used activation functions.
In \cite{ruff2018deep}, is has been shown that using bias terms, bounded activation functions or a trainable center all enable the model to learn a constant function when using an additive weight decay regularization term and thus must also be avoided.
\par
Another possibility to avoid trivial solutions is to impose additional tasks, so-called \emph{auxiliary tasks}, not directly related to the \ac{asd} problem while training.
Autoencoders \cite{hinton2006reducing}, which are trained to first encode and then decode the input again and have many interesting applications by themselves such as denoising data \cite{vincent2010stacked}, can also be viewed as a way to regularize one-class models.
Here, the encoder is the one-class model mapping the input to an embedding space.
Learning a constant function is not a (trivial) solution for the task because all necessary information for being able to completely reconstruct the input needs to be encoded.
However, noise including other sound sources present in the input audio data needs to be encoded as well because otherwise the input cannot be reconstructed.
Therefore, the noise heavily influences the embeddings and thus the embeddings can also be considered noisy.
Depending on the complexity of the noise, most information contained in the embeddings is only related to the noise and not to the target sound to be analyzed and thus detecting anomalies using an autoencoder may be difficult.
Moreover, in \cite{ruff2018deep} it has been shown that using compactness loss, even for clean datasets, outperforms commonly used autoencoder architectures when detecting anomalies.
\par
A second choice of an auxiliary task to prevent the model from learning a constant function as a trivial solution is a classification task.
Defining multiple classes through an auxiliary task inherently prevents learning a constant function as this would not be a (trivial) solution to the imposed classification problem.
In \cite{perera2019learning}, an additional \emph{descriptiveness loss} is used whose goal is to reduce inter-class similarity between classes of an arbitrary, external multi-class dataset, which is only used to regularize the one-class classification task.
This is done by minimizing the standard \ac{cce} loss for classification on this additional dataset as an auxiliary task.
For each of the two tasks, another version of the same network with identical structure and tied weights is used.
During training, both losses are jointly minimized using a weighted sum ensuring that the so-called reference network associated with the compactness loss does not learn a constant function because this would prevent the secondary network to be able to classify correctly.
\begin{rem}
The original definition of the compactness loss \cite{ruff2018deep} also includes an additional weight decay term.
Such a weight decay term can be used to complement any loss function and does not prevent the model from learning trivial solutions as it is still possible that the model learns to map everything to the center.
Furthermore, all theoretical results presented in this work are valid regardless of whether this specific weight decay term is included or not.
The proof of the main theorem can easily be modified to including the same weight decay term because it is just an additional additive term.
Therefore, we omitted this term in the theoretical investigations of this work for the sake of simplicity while still using it in our experiments.
However, we did not notice any significant effect on the performance.
\end{rem}
For the remainder of this work, we propose to normalize all representations in the embedding space $\mathbb{R}^D$, meaning that ${\lVert c\rVert_2=1=\lVert \phi(x,w)\rVert_2}$ for all $x\in X,w\in W(\phi)$ and centers ${c\in\mathbb{R}^D}$.
This can easily be achieved by dividing the embeddings by their corresponding Euclidean norms.
A normalization of the embedding space essentially reduces the dimension by one as evident by using stereographic projection.
But doing so does not degrade the \ac{asd} performance because the dimension of the embedding space usually is larger than it needs to be.
\par
Normalizing the embedding space has several advantages.
Most importantly, the initialization of the centers is substantially simplified.
In high-dimensional vector spaces i.i.d. random elements are almost surely approximately orthogonal \cite{gorban2016approximation}.
Hence, all class centers can be randomly initialized by sampling from a uniform random distribution as also done in \cite{wilkinghoff2023design} and a careful strategy for initializing the class centers is not needed.
This does not cause any problems e.g. by accidentally using class centers that are very similar to each other in terms of cosine similarity whereas the corresponding acoustic classes are very dissimilar or vice versa.
Moreover, normalizing the centers ensures that all centers are distributed equidistantly and sufficiently far away from zero to avoid learning a trivial solution.
Last but not least, normalizing the embeddings may even prevent numerical issues while training similar to when using batch normalization \cite{ioffe2015batch}.

\subsection{Angular Margin Losses}
We will review the definition of ArcFace \cite{deng2019arcface} as a representative of angular margin losses.
\begin{defn}[ArcFace]
\label{def:arcface}
Let $Y\subset X$ be finite and ${l_j(x)\in\lbrace0,1\rbrace}$ denote the $j$th component of the categorical class label function $l\in L$ where $L$ denotes the space of all functions $l:X\rightarrow\lbrace0,1\rbrace^N$ with $\sum_{j=1}^Nl_j(x)=1$ for all $x\in X$.
Let $\mathcal{P}$ denote the power set, $\Phi$ denote the space of network architectures for extracting embeddings and $W(\phi)$ denote the parameter space of $\phi\in\Phi$, thus $\phi:X\times W(\phi)\rightarrow\mathbb{R}^D$.
Let $\softmax:\mathbb{R}^N\rightarrow[0,1]^N$ denote the softmax function, i.e.
\BE\softmax(x)_i=\frac{\exp(x_i)}{\sum_{j=1}^N\exp(x_j)}.\EE
Then, the \emph{ArcFace} loss is defined as
\BE&\mathcal{L}_\text{ang}:\mathcal{P}(X)\times \mathcal{P}(\mathbb{R}^{D}) \times\Phi\times W\times L\times\mathbb{R}_+\times[0,\frac{\pi}{2}]\rightarrow\mathbb{R}_+\\
&\mathcal{L}_\text{ang}(Y,C,\phi,w,l,s,m)\\:=&-\frac{1}{\lvert Y\rvert}\sum_{x\in Y}\sum_{j=1}^Nl_j(x)\log(\softmax(s\cdot\cos_{\text{mar}}(\phi(x,w),c_j,m)))\EE
where $\lvert C\rvert=N$ and, in this case,
\BE&\softmax(s\cdot\cos_{\text{mar}}(\phi(x,w),c_i,m))\\:=&\frac{\exp(s\cdot\cos_{\text{mar}}(\phi(x,w),c_i,m))}{\sum_{j=1}^N \exp(s\cdot\cos_{\text{mar}}(\phi(x,w),c_j,m\cdot l_j(x))}\EE
with
\BE \cos_{\text{mar}}(x,y,m):=\cos(\arccos(\cos(x,y))+m)\EE 
for cosine similarity
\BE \cos(x,y):=\frac{\langle x,y\rangle}{\lVert x\rVert_2\lVert y\rVert_2}\in[-1,1].\EE
The vectors $c_j\in\mathbb{R}^D$ are called \emph{class centers}, $m\in[0,\frac{\pi}{2}]$ is called \textit{margin} and $s\in\mathbb{R}_+$ is called \textit{scale parameter}.
\end{defn}
\begin{rem}
When using mixup \cite{zhang2017mixup} for data augmentation, the definition of the class label function needs to be generalized to $l:X\rightarrow[0,1]^N$ with $\sum_{j=1}^Nl_j(x)=1$ for all $x\in X$.
In the experimental evaluations of this work, mixup will be used when training a model as this improves the \ac{asd} performance \cite{wilkinghoff2021sub}.
Furthermore, the theoretical results presented in this work still hold when using mixup but in the proofs only binary labels will be used for the sake of simplicity.
\end{rem}

In \cite{zhang2019adacos}, it has been shown that the choice of both hyperparameters, the scale parameter $s$ and the margin $m$, can have a significant impact on the resulting performance.
Strongly varying the magnitude of one of the individual parameters has a similar effect on the sensitivity of the posterior probabilities with respect to the angles as varying the other parameter.
Both a scale parameter that is too large and a margin that is too small lead to very high posterior probabilities for the target class, approximately equal to one, even for relatively large angles.
Therefore, the loss function is insensitive to changing the angle.
A scale parameter that is too small limits the maximum posterior probability of the target class that can be achieved.
Similarly, a margin that is too large also leads to relatively small posterior probabilities.
Thus, in both cases the model still tries to adapt its parameters even when the angles are already small, which hinders convergence.
Due to the similar behavior of both parameters, a single appropriately chosen parameter is sufficient for controlling the posterior probabilities and it has even been shown that an adaptive scale parameter outperforms using two tuned but fixed parameters.
Therefore, we will assume that $s$ is adaptive as specified for the AdaCos loss in \cite{zhang2019adacos} and set $m=0$, i.e. $\cos_\text{mar}(x,y,0)=\cos(x,y)$ for the remainder of this work.
Formally, the definition of the AdaCos loss is the following.

\begin{defn}[AdaCos]
\label{def:adacos}
Using the same notation as in Definition \ref{def:arcface}, let $Y^{(t)}\subset Y$ denote all samples belonging to a mini-batch of size $B\in\mathbb{N}$, i.e. $\lvert Y^{(t)}\rvert = B$.
Let ${\theta_{x,i}:=\arccos(\cos(\phi(x,w),c_i))\in[0,\pi]}$ and the \emph{dynamically adaptive scale parameter} $\tilde{s}^{(t)}\in\mathbb{R}_+$ at training step $t\in\mathbb{N}_0$ be set to
\BE\tilde{s}^{(t)}:=\begin{cases}
	\sqrt{2}\cdot\log(N-1) &\text{if }t=0\\
	\frac{\log B_\text{avg}^{(t)}}{\cos\big(\min(\frac{\pi}{4},\theta_\text{med}^{(t)})\big)}&\text{else}
\end{cases}\EE
where $\theta_\text{med}^{(t)}\in[0,\pi]$ denotes the median of all angles $\theta_{x,i(x)}$ with $x\in X^{(t)}$ and $i(x)\in\lbrace1,...,N\rbrace$ such that $l_i(x)=1$ and
\BE
B_\text{avg}^{(t)}:=\frac{1}{B}\sum_{x\in Y^{(t)}}\sum_{\substack{j=1\\ l_j(x)\neq1}}^N\exp\big(\tilde{s}^{(t-1)}\cdot\cos(\phi(x,w),c_j)\big)
\EE
is the sample-wise average over all summed logits belonging to the non-corresponding classes.
Then, the \emph{AdaCos} loss is defined as
\BE &\mathcal{L}_\text{ada}:\mathcal{P}(X)\times \mathcal{P}(\mathbb{R}^{D}) \times\Phi\times W\times L\rightarrow\mathbb{R}_+\\
&\mathcal{L}_\text{ada}(Y,C,\phi,w,l):=\mathcal{L}_\text{ang}(Y,C,\phi,w,l,\tilde{s},0).\EE
\end{defn}
\begin{rem}
When using mixup \cite{zhang2017mixup} for data augmentation, ${\theta_\text{med}^{(t)}\in[0,\pi]}$ needs to be replaced with the median of the mixed-up angles as specified in \cite{wilkinghoff2021sub}.
\end{rem}
The AdaCos loss can also be extended to using multiple centers for each class, called sub-clusters, instead of a single one.
The idea of using these sub-clusters is to allow the network to learn more complex distributions than a normal distribution for each class enabling the model to have a more differentiated view on the embeddings when using the cosine similarity as an anomaly score.
This has been shown to improve the \ac{asd} performance \cite{wilkinghoff2021sub} and thus helps to differentiate between normal and anomalous samples.
\begin{defn}[Sub-cluster AdaCos]
\label{def:sc_adacos}
Using the same notation as in Definitions \ref{def:arcface} and \ref{def:adacos}, let $C_j\in\mathcal{P}(\mathbb{R}^{D})$ with $\lvert C_j\rvert=M$ denote all centers belonging to class $j\in\lbrace1,...,N\rbrace$.
Let the \emph{dynamically adaptive scale parameter} $\hat{s}^{(t)}\in\mathbb{R}_+$ at training step $t\in\mathbb{N}_0$ be set to
\BE\hat{s}^{(t)}:=\begin{cases}
	\sqrt{2}\cdot\log(N\cdot M-1) &\text{if }t=0\\
	\frac{f_\text{max}^{(t)}+\log \hat{B}_\text{avg}^{(t)}}{\cos\big(\min(\frac{\pi}{4},\theta_\text{med}^{(t)})\big)}&\text{else}
\end{cases}\EE
with
\BE
\hat{B}_\text{avg}^{(t)}:=\frac{1}{B}\sum_{x\in Y^{(t)}}\sum_{j=1}^N\sum_{c\in C_j}\exp\big(\hat{s}^{(t-1)}\cos(\phi(x,w),c)-f_\text{max}^{(t)}\big)
\EE
and
\BE {f_\text{max}^{(t)}:=\max_{x\in Y^{(t)}}\max_{j=1}^N\max_{c\in C_j}\hat{s}^{(t-1)}\cdot\cos(\phi(x,w),c)}.\EE
Then, the \emph{sub-cluster AdaCos} loss is defined as
\BE &\mathcal{L}_\text{sc-ada}:\mathcal{P}(X)\times \mathcal{P}(\mathcal{P}(\mathbb{R}^{D})) \times\Phi\times W\times L\rightarrow\mathbb{R}_+\\
&\mathcal{L}_\text{sc-ada}(Y,C,\phi,w,l)\\:=&-\frac{1}{\lvert Y\rvert}\sum_{x\in Y}\sum_{j=1}^Nl_j(x)\log(\softmax(\hat{s}\cdot\cos(\phi(x,w),C_j)))\EE
where $\lvert C\rvert=N$ and, in this case,
\BE&\softmax(\hat{s}\cdot\cos(\phi(x,w),C_j))\\:=&\sum_{c_j\in C_j}\frac{\exp(\hat{s}\cdot\cos(\phi(x,w),c_j))}{\sum_{k=1}^{N}\sum_{c_k\in C_k} \exp(\hat{s}\cdot\cos(\phi(x,w),c_k))}\EE
\end{defn}
\begin{rem}
As shown in \cite{wilkinghoff2021sub}, for the sub-cluster AdaCos loss as defined above mixup \cite{zhang2017mixup} needs to be used. Otherwise, the dynamically adaptive scale parameter $\hat{s}^{(t)}$ grows exponentially.
\end{rem}
    For the compactness loss, there is no benefit of using sub-clusters.
    The reason is that an optimal solution of this sub-cluster compactness loss would correspond to the mean of the sub-clusters or, in case all embeddings are normalized, to its projection onto the unit sphere. 
    Hence, there would be a single global optimum and this sub-cluster compactness loss would behave as if only a single sub-cluster is used.
    For the sub-cluster AdaCos loss, the situation is completely different because the softmax function is applied to all individual sub-clusters and the sum over the resulting scores is taken.
    This makes the resulting softmax probability, and thus also the loss function, symmetric with respect to the corresponding sub-clusters of an individual class.
    Therefore, the loss is invariant to changing the position of an embedding on the hypersphere as long as the sum of the distances to the sub-clusters is the same.
    Hence, also the space of optimal solutions grows with respect to the number of sub-clusters.
    However, due to the dependence on the sub-clusters of the other classes caused by the softmax function, this invariance is a simplification and the real situation is more complex.

\section{Relation between One-Class Losses and Angular Margin Losses}
\label{sec:theory}
For the proof of the main theoretical result of this work, the following basic identity is needed.
\begin{lem}
\label{lem:cos_id}
    For $x,y\in\mathbb{R}^D$ with $\lVert x\rVert_2=\lVert y\rVert_2=1$, it holds that
    \BE\cos(x,y)=1-\frac{\lVert x-y\rVert_2^2}{2}.\EE
\end{lem}
\begin{proof}\renewcommand{\qedsymbol}{}
    See Appendix.
\end{proof}
\begin{rem}
This lemma also shows that for normalized embeddings using Euclidean distance and using cosine distance, which in this case is equal to the standard scalar product, are equivalent for computing an anomaly score.
\end{rem}
Now, the theorem itself follows.
\begin{thm}
\label{thm:ang_comp}
Let $Y_j:=\lbrace x\in Y:l_j(x)=1\rbrace$.
Then minimizing $\mathcal{L}_\text{sc-ada}(Y,C,\phi,w,l)$ with gradient descent minimizes all \ac{ic} compactness losses with weighted gradients given by
\BE &\frac{\hat{s}}{2}\sum_{i=1}^N\frac{1}{\lvert Y_i\rvert}\sum_{x\in Y_i}\sum_{c_i\in C_i}P(\tau(\phi(x,w))=c_i\vert \tau(\phi(x,w))\in C_i)\\\cdot&\frac{\partial}{\partial w}\lVert \phi(x,w)-c_i\rVert_2^2\EE
while maximizing all inter-class compactness losses with weighted gradients given by
\BE &-\frac{\hat{s}}{2}\sum_{i=1}^N\frac{1}{\lvert Y_i\rvert}\sum_{x\in Y_i}\sum_{k=1}^N\sum_{c_k\in C_k}P(\tau(\phi(x,w))=c_k)\\\cdot&\frac{\partial}{\partial w}\lVert \phi(x,w)-c_k\rVert_2^2\EE
where
\BE &P(\tau(\phi(x,w))=c_i\vert \tau(\phi(x,w))\in C_i)\\:=&\frac{\exp(\hat{s}\cdot\cos(\phi(x,w),c_i))}{\sum_{c'_i\in C_i}\exp(\hat{s}\cdot\cos(\phi(x,w),c'_i))}\EE
and
\BE &P(\tau(\phi(x,w))=c_k)\\:=&\frac{\exp(\hat{s}\cdot\cos(\phi(x,w),c_k))}{\sum_{k=1}^N\sum_{c'_k\in C_k}\exp(\hat{s}\cdot\cos(\phi(x,w),c'_k))}\EE
with a cluster assignment function $\tau:\mathbb{R}^D\rightarrow\mathbb{R}^D$ given by
\BE \tau(z,C)=\argmax_{c\in C}\cos(z,c).\EE
\end{thm}
\begin{proof}
Let $x\in Y$, $\phi\in\Phi$ and $\hat{s}\in\mathbb{R}_+$ be fixed and $i\in\lbrace1,...,N\rbrace$ such that ${l_i(x)=1}$ and $l_j(x)=0$ for $j\neq i$.
To simplify notation, define ${e(w,c):=\exp(\hat{s}\cdot\cos(\phi(x,w),c))}$.
Using Lemma \ref{lem:cos_id}, we see that
\BEs&\frac{\partial}{\partial w}\log\bigg(\sum_{c_i\in C_i}e(w,c_i)\bigg)\\
=&\frac{\sum_{c_i\in C_i}e(w,c_i)\cdot \hat{s}\cdot\frac{\partial}{\partial w}\cos(\phi(x,w),c_i))}{\sum_{c'_i\in C_i}e(w,c'_i)}\\
=&-\frac{\hat{s}}{2}\sum_{c_i\in C_i}\frac{e(w,c_i)\cdot\frac{\partial}{\partial w}\lVert \phi(x,w)-c_i\rVert_2^2}{\sum_{c'_i\in C_i}e(w,c'_i)}\EEs
and similarly
\BEs&\frac{\partial}{\partial w}\log\bigg(\sum_{k=1}^N\sum_{c_k\in C_k}e(w,c_k)\bigg)\\
=&\frac{\sum_{k=1}^N\sum_{c_k\in C_k}e(w,c_k)\cdot \hat{s}\cdot\frac{\partial}{\partial w}\cos(\phi(x,w),c_k))}{\sum_{k=1}^N\sum_{c'_k\in C_k}e(w,c'_k)}\\
=&-\frac{\hat{s}}{2}\sum_{k=1}^N\sum_{c_k\in C_k}\frac{e(w,c_k)\cdot\frac{\partial}{\partial w}\lVert \phi(x,w)-c_k\rVert_2^2}{\sum_{k=1}^N\sum_{c'_k\in C_k}e(w,c'_k)}\\
=&-\frac{\hat{s}}{2}\sum_{c_i\in C_i}\frac{e(w,c_i)\cdot\frac{\partial}{\partial w}\lVert \phi(x,w)-c_i\rVert_2^2}{\sum_{k=1}^N\sum_{c'_k\in C_k}e(w,c'_k)}\\
&-\frac{\hat{s}}{2}\sum_{\substack{k=1\\k\neq i}}^N\sum_{c_k\in C_k}\frac{e(w,c_k)\cdot\frac{\partial}{\partial w}\lVert \phi(x,w)-c_k\rVert_2^2}{\sum_{k=1}^N\sum_{c'_k\in C_k}e(w,c'_k)}.\EEs
Using both identities, we obtain
\BEs
&\frac{\partial}{\partial w}\sum_{j=1}^Nl_j(x)\log(\softmax(\hat{s}\cdot\cos(\phi(x,w),C_j)))\\
=&\frac{\partial}{\partial w}\log\bigg(\sum_{c_i\in C_i}\frac{e(w,c_i)}{\sum_{k=1}^N\sum_{c_k\in C_k}e(w,c_k)}\bigg)\\
=&\frac{\partial}{\partial w}\log\bigg(\sum_{c_i\in C_i}e(w,c_i)\bigg)-\frac{\partial}{\partial w}\log\bigg(\sum_{k=1}^N\sum_{c_k\in C_k}e(w,c_k)\bigg)\\
=&-\frac{\hat{s}}{2}\sum_{c_i\in C_i}\frac{e(w,c_i)\cdot\frac{\partial}{\partial w}\lVert \phi(x,w)-c_i\rVert_2^2}{\sum_{c'_i\in C_i}e(w,c'_i)}\\
&+\frac{\hat{s}}{2}\sum_{c_i\in C_i}\frac{e(w,c_i)\cdot\frac{\partial}{\partial w}\lVert \phi(x,w)-c_i\rVert_2^2}{\sum_{k=1}^N\sum_{c'_k\in C_k}e(w,c'_k)}\\
&+\frac{\hat{s}}{2}\sum_{\substack{k=1\\k\neq i}}^N\sum_{c_k\in C_k}\frac{e(w,c_k)\cdot\frac{\partial}{\partial w}\lVert \phi(x,w)-c_k\rVert_2^2}{\sum_{k=1}^N\sum_{c'_k\in C_k}e(w,c'_k)}\\
=&-\frac{\hat{s}}{2}\bigg(\sum_{c_i\in C_i}e(w,c_i)\cdot\frac{\partial}{\partial w}\lVert \phi(x,w)-c_i\rVert_2^2\\
&\cdot\bigg(\frac{1}{\sum_{c'_i\in C_i}e(w,c'_i)}-\frac{1}{\sum_{k=1}^N\sum_{c'_k\in C_k}e(w,c'_k)}\bigg)\\
&-\sum_{\substack{k=1\\k\neq i}}^N\sum_{c_k\in C_k}\frac{e(w,c_k)\cdot\frac{\partial}{\partial w}\lVert \phi(x,w)-c_k\rVert_2^2}{\sum_{k=1}^N\sum_{c'_k\in C_k}e(w,c'_k)}\bigg)\\
=&-\frac{\hat{s}}{2}\bigg(\sum_{c_i\in C_i}e(w,c_i)\cdot\frac{\partial}{\partial w}\lVert \phi(x,w)-c_i\rVert_2^2\\
&\cdot\bigg(\sum_{\substack{k=1\\k\neq i}}^N\sum_{c_k\in C_k}\frac{e(w,c_k)}{(\sum_{c'_i\in C_i}e(w,c'_i))(\sum_{k=1}^N\sum_{c'_k\in C_k}e(w,c'_k))}\bigg)\\
&-\sum_{\substack{k=1\\k\neq i}}^N\sum_{c_k\in C_k}\frac{e(w,c_k)\cdot\frac{\partial}{\partial w}\lVert \phi(x,w)-c_k\rVert_2^2}{\sum_{k=1}^N\sum_{c'_k\in C_k}e(w,c'_k)}\bigg)\\
=&-\frac{\hat{s}}{2}\sum_{\substack{k=1\\k\neq i}}^N\sum_{c_k\in C_k}\frac{e(w,c_k)}{\sum_{k=1}^N\sum_{c'_k\in C_k}e(w,c'_k)}\\
&\cdot\bigg(\sum_{c_i\in C_i}\frac{e(w,c_i)}{\sum_{c'_i\in C_i}e(w,c'_i)}\cdot\frac{\partial}{\partial w}\lVert \phi(x,w)-c_i\rVert_2^2\\&-\frac{\partial}{\partial w}\lVert \phi(x,w)-c_k\rVert_2^2\bigg)\\
=&-\frac{\hat{s}}{2}\sum_{k=1}^N\sum_{c_k\in C_k}\underbrace{\frac{e(w,c_k)}{\sum_{k=1}^N\sum_{c'_k\in C_k}e(w,c'_k)}}_{=P(\tau(\phi(x,w))=c_k)}\\&\cdot\sum_{c_i\in C_i}\underbrace{\frac{e(w,c_i)}{\sum_{c'_i\in C_i}e(w,c'_i)}}_{=P(\tau(\phi(x,w))=c_i\vert \tau(\phi(x,w))\in C_i)}\\
&\cdot\bigg(\frac{\partial}{\partial w}\lVert \phi(x,w)-c_i\rVert_2^2-\frac{\partial}{\partial w}\lVert \phi(x,w)-c_k\rVert_2^2\bigg)\EEs
where we used that
\BEs&\frac{1}{\sum_{c'_i\in C_i}e(w,c'_i)}-\frac{1}{\sum_{k=1}^N\sum_{c'_k\in C_k}e(w,c'_k)}\\
=&\frac{\sum_{k=1}^N\sum_{c_k\in C_k}e(w,c_k)-\sum_{c_i\in C_i}e(w,c_i)}{(\sum_{c'_i\in C_i}e(w,c'_i))(\sum_{k=1}^N\sum_{c'_k\in C_k}e(w,c'_k))}\\
=&\sum_{\substack{k=1\\k\neq i}}^N\sum_{c_k\in C_k}\frac{e(w,c_k)}{(\sum_{c'_i\in C_i}e(w,c'_i))(\sum_{k=1}^N\sum_{c'_k\in C_k}e(w,c'_k))}.\EEs
Now, summing over all samples $x\in Y$, normalizing with $\lvert Y\rvert$ and taking the additive inverse yields the desired result.
\par
When using mixup, the right hand side of the last equation needs to be replaced with a weighted sum of two terms, each corresponding to one of the two classes that are mixed-up, because there are ${i_1,i_2\in\lbrace1,...,N\rbrace}$ such that ${l_{i_1}(x)\neq0\neq l_{i_2}(x)}$.
Otherwise, the proof is exactly the same.
In conclusion, the proven result still holds for mixed-up samples but includes two similar terms instead of one term.
\end{proof}
\begin{cor}
Minimizing $\mathcal{L}_\text{ada}(Y,C,\phi,w,l)$ with gradient descent is equivalent to minimizing 
\BE &-\frac{\tilde{s}}{2}\sum_{\substack{k=1}}^N\softmax(\hat{s}\cdot\cos(\phi(x,w),c_k))\\
\cdot&\bigg(\frac{\partial}{\partial w}\lVert \phi(x,w)-c_i\rVert_2^2-\frac{\partial}{\partial w}\lVert \phi(x,w)-c_k\rVert_2^2\bigg).\EE
\end{cor}

\begin{proof}
The proof of Theorem \ref{thm:ang_comp} does not depend on the exact structure of the dynamically adaptive scale parameter and thus also holds for the standard AdaCos loss by replacing $\hat{s}$ with $\tilde{s}$ and using only a single sub-cluster for each class.
\end{proof}

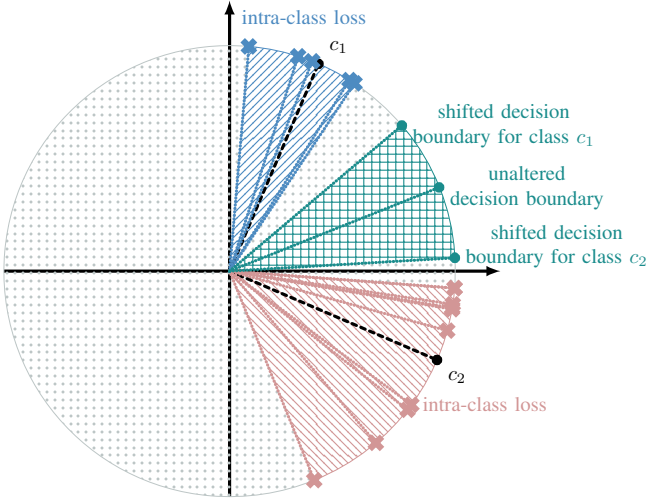
\begin{figure}[t]
    \centering
    \begin{adjustbox}{width=\columnwidth}
          \begin{tikzpicture}[very thick, scale=1, every node/.style={scale=1.3}, line width=2]
  \tkzInit[xmin=-1,xmax=1.1,xstep=.2,ymin=-1,ymax=1.1,ystep=.2]
  \tkzDrawX[label=]
  \tkzDrawY[label=]

  \tkzDefPoint(0,0){O}
  \tkzDefPoint(0.3939193, 0.91914503){A}
  \tkzDefPoint(0.91914503, -0.3939193){B}
  \tkzDefPoint(0.32989211, 0.94401864){AM}
  \tkzDefPoint(0.80753613, -0.58981811){BM}
  
  \tkzDefPoint(0.55371146, 0.8327086){A1}
  \tkzDefPoint(0.36614142, 0.93055922){A2}
  \tkzDefPoint(0.53128191, 0.8471951){A3}
  \tkzDefPoint(0.30374254, 0.95275415){A4}
  \tkzDefPoint(0.08545945, 0.99634165){A5}

  \tkzDefPoint(0.99715266, -0.07540939){B1}
  \tkzDefPoint(0.37519785, -0.92694475){B2}
  \tkzDefPoint(0.96456098, -0.26386002){B3}
  \tkzDefPoint(0.64817275, -0.76149333){B4}
  \tkzDefPoint(0.99715266, -0.07540939){B5}
  \tkzDefPoint(0.80563177, -0.59241662){B6}
  \tkzDefPoint(0.98561993, -0.16897739){B7}
  \tkzDefPoint(0.989149 , -0.1469158){B8}
  \tkzDefPoint(0.98663952, -0.16291855){B9}
  \tkzDefPoint(0.78857062, -0.61494421){B10}

  \tkzDefPoint(0.92847669, 0.37139067){M}
  \tkzDefPoint(0.99827437, 0.05872202){MB}
  \tkzDefPoint(0.76338629, 0.64594224){MA}

  \begin{scope}
    \tkzClipCircle(O,A)
  \end{scope}

  \tkzDrawArc[color=teal!15!gray!45](O,A5)(B2)
  \tkzDrawArc[color=teal!15!gray!45](O,MA)(A1)
  \tkzDrawArc[color=teal!15!gray!45](O,B5)(MB)
  \tkzDrawArc[latex-latex,color=teal!90](O,MB)(MA)
  \tkzDrawArc[latex reversed-latex reversed,color=teal!70!blue!70](O,A1)(A5)
  \tkzDrawArc[latex reversed-latex reversed,color=teal!35!red!50](O,B2)(B5)
  \tkzFillAngle[fill=blue,pattern=north east lines,pattern color=teal!70!blue!70,size=5cm](A1,O,A5)
  \tkzFillAngle[fill=blue,pattern=north west lines,pattern color=teal!35!red!50,size=5cm](B2,O,B5)
  \tkzFillAngle[fill=blue,pattern=dots,pattern color=teal!15!gray!45,size=5cm](A5,O,B2)
  \tkzFillAngle[fill=blue,pattern=dots,pattern color=teal!15!gray!45,size=5cm](MA,O,A1)
  \tkzFillAngle[fill=blue,pattern=dots,pattern color=teal!15!gray!45,size=5cm](B5,O,MB)
  \tkzFillAngle[fill=blue,pattern=grid,pattern color=teal!90,size=5cm](MB,O,MA)
  \tkzDrawSegments[dashed,line width=2](O,A)
  \tkzDrawSegments[dashed,line width=2](O,B)
  \tkzDrawSegments[dotted,color=teal!35!red!50,line width=2](O,B1 O,B2 O,B3 O,B4 O,B5 O,B6 O,B7 O,B8 O,B9 O,B10)
  \tkzDrawSegments[dotted,color=teal!70!blue!70,line width=2](O,A1 O,A2 O,A3 O,A4 O,A5)
  \tkzDrawPoints[line width=4](A,B)
  \tkzDrawPoints[shape=cross out,color=teal!70!blue!70,line width=4](A1,A2,A3,A4,A5)
  \tkzDrawPoints[shape=cross out,color=teal!35!red!50,line width=4](B1,B2,B3,B4,B5,B6,B7,B8,B9,B10)

  \tkzLabelPoint[above right](A){$c_1$}
  \tkzLabelPoint[below right](B){$c_2$}
  \tkzDrawPoints[line width=4,color=teal!90](MB)
  \tkzDrawSegments[dotted,line width=2,color=teal!90](O,MB)
  \tkzLabelPoint[right,align=center,color=teal!90,yshift=0.2cm](MB){shifted decision\\ boundary for class $c_2$}
  \tkzDrawPoints[line width=4,color=teal!90](MA)
  \tkzDrawSegments[dotted,line width=2,color=teal!90](O,MA)
  \tkzLabelPoint[right,align=center,color=teal!90](MA){shifted decision\\ boundary for class $c_1$}
  \tkzDrawPoints[line width=4,color=teal!90](M)
  \tkzDrawSegments[dotted,line width=2,color=teal!90](O,M)
  \tkzLabelPoint[right,align=center,color=teal!90](M){unaltered\\ decision boundary}
  \tkzLabelPoint[above,color=teal!70!blue!70,yshift=0.4cm](AM){{intra-class loss}}
  \tkzLabelPoint[right,color=teal!35!red!50](BM){{intra-class loss}}
\end{tikzpicture}
    \end{adjustbox}
    \caption{Illustration of \ac{ic} compactness losses and the angular margin to be ensured between the classes for ${D=2,N=2,M=1}$.
    Intra-class losses are computed by summing all distances of samples to their corresponding class centers (blue and red areas). Inter-class losses are computed by summing all distances of samples to their corresponding decision boundaries. An unaltered decision boundary is exactly the midpoint between the class centers. When using an angular margin loss, the decision boundaries to the other classes are essentially shifted closer to the class center for which the inter-class loss is computed (see Fig. 1 in \cite{wang2018additive}). This explicitly ensures a margin between the classes, which is depicted by the green area.}
    \label{fig:unit_circle}
\end{figure}

This theorem shows that using an angular margin loss such as the AdaCos loss is essentially the same strategy as proposed in \cite{perera2019learning} and applied to \ac{asd} in \cite{inoue2020detection}, i.e. using a compactness loss for increasing \ac{ic} similarity, as defined in Definition \ref{def:comp}, and a so-called descriptiveness loss to decrease inter-class similarity.
However, there are differences between both approaches.
When minimizing an angular margin loss, inter-class compactness losses are used to decrease inter-class similarity instead of a standard \ac{cce} loss.
Second, when using two loss functions one usually has to tune a weight parameter to create a weighted sum of both loss terms, which is not needed for an angular margin loss and impossible without access to anomalous samples.
Furthermore, the gradients belonging to individual samples are weighted with specific softmax probabilities giving more emphasis the closer the sub-clusters are.
As these weights are non-uniform in general, this explicitly shows why using multiple sub-clusters is not equivalent to using a single sub-cluster given by the projection of the mean of the sub-clusters onto the hypersphere as it is the case for an \ac{ic} compactness loss with multiple sub-clusters.
Last but not least, an angular margin loss explicitly ensures a margin between classes, as illustrated in Fig. \ref{fig:unit_circle}, whereas a combination of compactness losses and a \ac{cce} loss only implicitly does this by increasing intra-class similarity.
Note that, in \cite{perera2019learning}, inter-class similarity is decreased on another dataset using less relevant classes because only a single class is available on the target dataset.
Because of these differences, directly minimizing an angular margin loss leads to a different solution than minimizing a combination of \ac{ic} losses and a descriptiveness loss.
\par
Note that the \ac{ic} compactness loss with multiple classes can also be considered a prototypical loss \cite{snell2017prototypical} or angular prototypical loss \cite{chung2020defence} as used for few-shot classification \cite{wang2021generalizing_few}, which defines settings where only very few training samples, called shots, are available for each class.
The only difference between these prototypical losses and an angular margin loss is that, for prototypical losses, the center vectors are re-calculated as the means of embeddings belonging to corresponding classes by using a so-called support set during training while, for an angular margin loss, the class centers are fixed or adaptable parameters of the network.
Hence, this theorem also shows that angular margin losses are a suitable choice for few-shot classification as shown for open-set sound event classification \cite{wilkinghoff2023using} and few-shot keyword spotting \cite{wilkinghoff2023tacos}.
\par
Choosing a classification task as an auxiliary task prevents learning a constant function as a trivial solution.
The reason is that, for such a classification task, an optimal solution is a classifier that maps each sample to its corresponding class center and thus corresponds to jointly learning multiple trivial solutions, one for each class, instead of only learning a constant function.
As long as each anomalous sample belongs to a well-defined normal class used during training, this optimal solution would yield representations not suitable for detecting anomalies as they would not be distinguishable from representations obtained with normal samples. 
However, obtaining such a perfect classifier is much more difficult than learning a constant mapping for a single class and thus training a single model to classify between multiple classes already prevents trivial solutions as long as the classification problem itself is not trivial e.g. by consisting of only a single class.
Still, in \cite{wilkinghoff2023design} it has been shown that the \ac{asd} performance can be improved by applying the same three strategies as used for the compactness loss \cite{ruff2018deep}, namely 1) not using bias terms, 2) not using bounded activation functions and 3) not using trainable class centers.
The most likely reason is that these strategies prevent the model to learn trivial solutions, leading to less informative embeddings, for individual classes that are easily recognized.

\section{Experimental Results}
\label{sec:exp}
Using one-class losses and angular margin losses for \ac{asd} will now be compared experimentally.
\subsection{Dataset}
\begin{table}[t]
	\centering
	\caption{Structure of the DCASE2022 \ac{asd} dataset 
 }
\begin{adjustbox}{max width=\columnwidth}
	\begin{tabular}{llrrrr}
		\toprule
        &&\multicolumn{4}{c}{number of recordings (per section)}\\
        subset&split&\multicolumn{2}{c}{source domain}&\multicolumn{2}{c}{target domain}\\
        &&normal&anomalous&normal&anomalous\\
		\midrule
		development&training&$990$&$0$&$10$&$0$\\
		development&test&$50$&$50$&$50$&$50$\\
		evaluation&training&$990$&$0$&$10$&$0$\\
		evaluation&test&$50$&$50$&$50$&$50$\\
		\bottomrule
	\end{tabular}
\end{adjustbox}
\label{tab:dataset}
\end{table}
For most experiments conducted in this work, the DCASE2022 \ac{asd} dataset \cite{dohi2022description} of the task titled \enquote{Unsupervised Anomalous Sound Detection for Machine Condition Monitoring Applying Domain Generalization Techniques} has been used.
The dataset consists of recordings of machine sounds with background factory noise.
Each recording has a single channel, a length of ten seconds and a sampling rate of $16$ kHz and belongs to one of the seven machine types \enquote{fan}, \enquote{gearbox}, \enquote{bearing}, \enquote{slide rail}, \enquote{valve} from MIMII DG \cite{dohi2022mihiidg} and \enquote{toy car}, \enquote{toy train} from ToyADMOS2 \cite{harada2021toyadmos2}.
For each machine type, there are six different so-called sections each of which is dedicated to a specific type of domain shift.
A domain shift means that the characteristics of a machine sound differ in some way between a source domain with many training samples and a target domain with only few training samples.
These shifts can be caused by physical changes of the machines e.g. caused by replacing parts for maintenance, or changes in the acoustical environment e.g. a different background noise or using different recording devices.
Ideally, the \ac{asd} system is able to reliably detect anomalies despite these domain shifts without the need for adapting the system (domain generalization \cite{wang2021generalizing}).
\par
The dataset is divided into a development and an evaluation split each containing recordings of $21$ sections, three for each machine type.
For each recording, information about the machine type and section are given.
For the training datasets, domain information (\enquote{source} or \enquote{target}) and additional attribute information such as states of machine types or noise conditions are given for each recording.
For the test datasets, no domain information and no additional attribute information are given.
The exact structure of the dataset can be found in Tab \ref{tab:dataset}.
The task of an \ac{asd} system is to reliably detect anomalous samples regardless of whether a sample belongs to a source or target domain, i.e. using a single decision threshold for both domains of a section.
\par
Some of the experiments have also been conducted on the DCASE2023 \ac{asd} dataset \cite{dohi2023description,harada2023first} belonging to the task \enquote{First-Shot Unsupervised Anomalous Sound Detection for Machine Condition Monitoring}.
Similar to the DCASE2022 \ac{asd} dataset, this dataset is also aimed at domain generalization for \ac{asd} with the following differences.
First and foremost, the development and evaluation split of the dataset contain different machine types.
The development set contains the same machine types as the DCASE2022 dataset, namely \enquote{fan}, \enquote{gearbox}, \enquote{bearing}, \enquote{slide rail}, \enquote{valve} from MIMII DG \cite{dohi2022mihiidg} and \enquote{toy car}, \enquote{toy train} from ToyADMOS2 \cite{harada2021toyadmos2}.
The evaluation set contains seven completely different machine types, namely \enquote{toy drone}, \enquote{n-scale toy train}, \enquote{vacuum}, and \enquote{toy tank} from \cite{harada2023toyadmos2+} and \enquote{bandsaw}, \enquote{grinder}, \enquote{shaker} from \cite{dohi2022mihiidg}.
Furthermore, for each machine type there is only a single section.
This lowers the difficulty of the auxiliary classification task and thus makes it more difficult to extract embeddings, which are sensitive to anomalous changes of the target sounds.
\par
For the DCASE \ac{asd} datasets, two performance measures are used to evaluate the performance of individual \ac{asd} systems.
One metric is the \ac{auc}, the other metric is the \ac{pauc} \cite{mcclish1989analyzing}, which is the AUC calculated over a low false positive rate ranging from $0$ to $p$ with $p=0.1$ in this case.
The \ac{pauc} is used as an additional metric because decision thresholds for machine condition monitoring are usually set to a value that gives a low number of false alarms and thus this area of the \ac{roc} curve is of particular interest.
Both are threshold-independent metrics allowing a more objective comparison between different \ac{asd} systems than threshold-dependent metrics \cite{aggarwal2017outlier,ebbers2022threshold}.
\begin{table*}[t]
	\centering
	\caption{\ac{asd} performance obtained with different losses using different auxiliary tasks. Harmonic means of all AUCs and pAUCs over all pre-defined sections of the dataset are depicted in percent. Arithmetic mean and standard deviation of the results over five independent trials are shown. Best results in each column are highlighted with bold letters.}
\begin{adjustbox}{max width=\textwidth}
	\begin{tabular}{llllllll}
		\toprule
        \multicolumn{8}{c}{DCASE2022 development set}\\
        \midrule
		\multirow{2}{*}{loss}&\multirow{2}{*}{classes of auxiliary task (number of classes)}&\multicolumn{2}{c}{source domain}&\multicolumn{2}{c}{target domain}&\multicolumn{2}{c}{both domains}\\
		&&AUC&pAUC&AUC&pAUC&AUC&pAUC\\
		\midrule
		\ac{ic} compactness loss (Def. \ref{def:comp})&none ($1$)&$56.4\pm1.4$&$53.9\pm0.6$&$53.6\pm0.9$&$52.6\pm0.3$&$55.1\pm1.2$&$52.6\pm0.4$\\
		\ac{ic} compactness loss (Def. \ref{def:comp})&machine types ($7$)&$66.5\pm2.9$&$60.6\pm0.6$&$63.6\pm2.2$&$57.1\pm0.9$&$65.0\pm1.7$&$57.8\pm0.6$\\
		\ac{ic} compactness loss (Def. \ref{def:comp})&machine types and sections ($42$)&$77.6\pm1.7$&$70.5\pm0.9$&$75.3\pm0.9$&$63.3\pm0.8$&$76.4\pm0.9$&$63.5\pm0.6$\\
		\ac{ic} compactness loss (Def. \ref{def:comp})&machine types and sections, models trained individually ($1$)&$50.0\pm2.3$&$52.1\pm0.6$&$51.7\pm1.8$&$52.2\pm0.4$&$51.8\pm1.8$&$51.4\pm0.4$\\
		\ac{ic} compactness loss (Def. \ref{def:comp})&machine types, sections and attribute information ($342$)&$80.7\pm1.9$&$73.7\pm1.0$&$74.5\pm0.9$&$62.1\pm1.2$&$78.1\pm0.8$&$63.3\pm0.9$\\
		\ac{ic} compactness loss (Def. \ref{def:comp}) + \ac{cce}&machine types, sections and attribute information ($342$)&$82.5\pm0.7$&$75.2\pm0.7$&$75.5\pm0.6$&$61.2\pm1.6$&$79.0\pm0.6$&$64.8\pm0.9$\\
		AdaCos loss (Def. \ref{def:adacos})&machine types, sections and attribute information ($342$)&$83.0\pm1.3$&$75.2\pm1.8$&$75.4\pm1.0$&$60.9\pm0.8$&$79.2\pm0.9$&$64.3\pm0.7$\\
		sub-cluster AdaCos loss (Def. \ref{def:sc_adacos})&machine types, sections and attribute information ($342$)&$\pmb{84.2\pm0.8}$&$\pmb{76.5\pm0.9}$&$\pmb{78.5\pm0.9}$&$\pmb{62.5\pm0.9}$&$\pmb{81.4\pm0.7}$&$\pmb{66.6\pm0.9}$\\
		\midrule
        \multicolumn{8}{c}{DCASE2022 evaluation set}\\
		\midrule
		\multirow{2}{*}{loss}&\multirow{2}{*}{classes of auxiliary task (number of classes)}&\multicolumn{2}{c}{source domain}&\multicolumn{2}{c}{target domain}&\multicolumn{2}{c}{both domains}\\
		&&AUC&pAUC&AUC&pAUC&AUC&pAUC\\
        \midrule
		\ac{ic} compactness loss (Def. \ref{def:comp})&none ($1$)&$49.9\pm0.8$&$50.6\pm0.4$&$51.0\pm0.4$&$51.0\pm0.7$&$50.9\pm0.5$&$50.3\pm0.4$\\
		\ac{ic} compactness loss (Def. \ref{def:comp})&machine types ($7$)&$59.6\pm1.3$&$56.9\pm0.5$&$57.6\pm1.8$&$53.8\pm0.9$&$59.3\pm1.5$&$54.6\pm0.6$\\
		\ac{ic} compactness loss (Def. \ref{def:comp})&machine types and sections ($42$)&$70.8\pm1.2$&$62.1\pm0.7$&$61.7\pm0.8$&$55.4\pm1.0$&$66.3\pm0.6$&$56.5\pm0.4$\\
		\ac{ic} compactness loss (Def. \ref{def:comp})&machine types and sections, models trained individually ($1$)&$52.9\pm1.4$&$51.7\pm0.5$&$54.5\pm0.6$&$51.6\pm0.3$&$54.2\pm0.8$&$51.2\pm0.3$\\
		\ac{ic} compactness loss (Def. \ref{def:comp})&machine types, sections and attribute information ($342$)&$73.7\pm0.5$&$63.4\pm0.7$&$67.9\pm1.0$&$57.8\pm1.3$&$70.9\pm0.6$&$58.5\pm0.9$\\
		\ac{ic} compactness loss (Def. \ref{def:comp}) + \ac{cce}&machine types, sections and attribute information ($342$)&$74.7\pm0.7$&$64.9\pm1.1$&$69.2\pm0.7$&$59.8\pm1.3$&$71.9\pm0.6$&$59.5\pm1.0$\\
		AdaCos loss (Def. \ref{def:adacos})&machine types, sections and attribute information ($342$)&$76.3\pm1.0$&$\pmb{66.0\pm0.5}$&$\pmb{69.9\pm0.8}$&$\pmb{59.9\pm1.5}$&$73.2\pm0.4$&$\pmb{60.1\pm0.9}$\\
		sub-cluster AdaCos loss (Def. \ref{def:sc_adacos})&machine types, sections and attribute information ($342$)&$\pmb{76.8\pm0.8}$&$65.8\pm0.2$&$69.8\pm0.5$&$59.7\pm1.1$&$\pmb{73.4\pm0.5}$&$59.8\pm0.8$\\
        \midrule
        \multicolumn{8}{c}{DCASE2023 development set}\\
		\midrule
		\multirow{2}{*}{loss}&\multirow{2}{*}{classes of auxiliary task (number of classes)}&\multicolumn{2}{c}{source domain}&\multicolumn{2}{c}{target domain}&\multicolumn{2}{c}{both domains}\\
		&&AUC&pAUC&AUC&pAUC&AUC&pAUC\\
        \midrule
		\ac{ic} compactness loss (Def. \ref{def:comp})&none ($1$)&$50.7\pm3.5$&$52.6\pm0.3$&$45.3\pm1.9$&$50.1\pm0.5$&$48.9\pm1.4$&$50.9\pm0.4$\\
		\ac{ic} compactness loss (Def. \ref{def:comp})&machine types ($14$)&$67.3\pm2.7$&$63.0\pm1.4$&$67.8\pm1.2$&\pmb{$58.6\pm1.1$}&$67.4\pm1.4$&\pmb{$59.4\pm1.1$}\\
		\ac{ic} compactness loss (Def. \ref{def:comp})&machine types, models trained individually ($1$)&$46.7\pm1.9$&$51.7\pm0.6$&$45.9\pm3.2$&$50.4\pm0.8$&$47.6\pm2.1$&$50.7\pm0.6$\\
		\ac{ic} compactness loss (Def. \ref{def:comp})&machine types and attribute information ($186$)&$67.6\pm2.5$&$61.6\pm1.2$&$70.0\pm2.4$&$56.4\pm1.9$&$68.3\pm1.9$&$57.1\pm1.3$\\
		\ac{ic} compactness loss (Def. \ref{def:comp}) + \ac{cce}&machine types and attribute information ($186$)&\pmb{$70.1\pm1.5$}&\pmb{$63.3\pm1.3$}&$71.0\pm1.3$&$55.5\pm1.1$&$70.4\pm1.0$&$56.7\pm0.8$\\
		AdaCos loss (Def. \ref{def:adacos})&machine types and attribute information ($186$)&$69.8\pm1.5$&$62.8\pm1.3$&$72.1\pm1.2$&$55.4\pm1.7$&\pmb{$71.2\pm0.7$}&$56.8\pm1.2$\\
		sub-cluster AdaCos loss (Def. \ref{def:sc_adacos})&machine types and attribute information ($186$)&$69.4\pm1.5$&$61.4\pm1.5$&\pmb{$72.4\pm1.6$}&$55.3\pm1.2$&$71.0\pm1.2$&$56.3\pm1.1$\\
        \midrule
        \multicolumn{8}{c}{DCASE2023 evaluation set}\\
		\midrule
		\multirow{2}{*}{loss}&\multirow{2}{*}{classes of auxiliary task (number of classes)}&\multicolumn{2}{c}{source domain}&\multicolumn{2}{c}{target domain}&\multicolumn{2}{c}{both domains}\\
		&&AUC&pAUC&AUC&pAUC&AUC&pAUC\\
        \midrule
		\ac{ic} compactness loss (Def. \ref{def:comp})&none ($1$)&$51.8\pm2.1$&$51.4\pm1.2$&$50.0\pm1.9$&$50.5\pm0.7$&$51.6\pm0.9$&$50.8\pm0.6$\\
		\ac{ic} compactness loss (Def. \ref{def:comp})&machine types ($14$)&$59.3\pm1.9$&$54.4\pm0.6$&$54.3\pm2.1$&$51.2\pm0.5$&$56.7\pm1.2$&$52.0\pm0.6$\\
		\ac{ic} compactness loss (Def. \ref{def:comp})&machine types, models trained individually ($1$)&$51.3\pm0.7$&$51.9\pm0.7$&$54.7\pm1.7$&$52.3\pm0.8$&$53.2\pm1.0$&$51.5\pm0.6$\\
		\ac{ic} compactness loss (Def. \ref{def:comp})&machine types and attribute information ($186$)&\pmb{$73.0\pm1.9$}&$62.1\pm1.4$&$58.9\pm2.7$&$55.1\pm1.4$&$64.1\pm1.8$&$55.6\pm0.8$\\
		\ac{ic} compactness loss (Def. \ref{def:comp}) + \ac{cce}&machine types and attribute information ($186$)&$72.6\pm1.4$&\pmb{$62.5\pm1.9$}&$62.2\pm2.6$&$56.2\pm0.8$&$67.2\pm0.7$&\pmb{$58.0\pm0.9$}\\
		AdaCos loss (Def. \ref{def:adacos})&machine types and attribute information ($186$)&$72.3\pm1.7$&$62.1\pm1.4$&$61.6\pm3.1$&\pmb{$56.4\pm1.1$}&$67.0\pm1.5$&$57.4\pm0.9$\\
		sub-cluster AdaCos loss (Def. \ref{def:sc_adacos})&machine types and attribute information ($186$)&$72.1\pm1.9$&$61.3\pm1.5$&\pmb{$62.3\pm2.7$}&$56.0\pm0.7$&\pmb{$67.3\pm1.3$}&$57.4\pm0.7$\\
		\bottomrule
	\end{tabular}
\end{adjustbox}
\label{tab:performances}
\end{table*}
\subsection{System Description}
The focus of this work is to explain why angular margin losses work well for \ac{asd}.
This requires using different loss functions for training an \ac{asd} system.
To this end, the conceptually simple state-of-the-art system presented in \cite{wilkinghoff2023design}, which only consists of a single model and uses the same settings for all machine types, is utilized.
For all experiments conducted in this work, only the loss function used for training the system is altered.
The system utilizes a magnitude spectrogram as well as the whole magnitude spectrum as input representations and uses two different convolutional sub-models for handling these, resulting in two different embeddings.
Then, both embeddings are concatenated to obtain a single embedding and the sub-cluster AdaCos loss \cite{wilkinghoff2021sub} is applied with $16$ sub-clusters, which are initialized uniformly at random, for training the model.
For the magnitude spectrogram, temporal mean normalization is applied to reduce the effect of different acoustic domains and make both input feature representations a bit more different by removing constant frequency information from the spectrograms.
Furthermore, the model does not use bias terms or trainable clusters as this improves the \ac{asd} performance by avoiding trivial solutions as discussed before.
The model is trained for $10$ epochs with a batch size of $64$ using mixup \cite{zhang2017mixup} with a uniform distribution for sampling the mixing coefficient and is implemented in Tensorflow \cite{abadi2016tensorflow}.
\par
After training the model using an auxiliary classification task, embeddings are extracted for the recordings.
For each section of the dataset, k-means with $k=16$ is applied to all normal training samples belonging to the source domain of this section.
The goal is to represent the distribution of the normal embeddings and be able to compute an anomaly score by taking the minimum cosine distance to the mean embeddings belonging to the same section as a given test sample.
Note that these means do not correspond to the sub-clusters as some sub-clusters may not have been used by the network during training.
It is possible that the embeddings are clustered between the sub-clusters due to the complex dependence between the sub-clusters of the other classes.
Still, it has been shown taking the same number of clusters usually performs best \cite{wilkinghoff2021sub}.
Since there are only $10$ normal samples available for the target domain, the minimum over the direct cosine distances to the corresponding embeddings is used.
As a last step, the minimum of the minimum cosine distances belonging to both domains is used to have an \ac{asd} system that generalizes to both domains.
Hence, a higher anomaly score indicates anomalous sounds whereas a smaller value indicates normal sounds.
More details about the system including a hyperlink to an open-source implementation can be found in \cite{wilkinghoff2023design}.

\subsection{Performance Evaluations}
Regardless of the loss function, training the \ac{asd} model without using anomalous samples is not directly targeting the \ac{asd} performance but only indirectly since the auxiliary task is aimed at obtaining embeddings suitable for \ac{asd}.
Although, there is a strong relation between the auxiliary and the \ac{asd} task, as otherwise training an \ac{asd} model by using an auxiliary task would not lead to usable representations, the actual \ac{asd} performance needs to be evaluated experimentally and cannot be investigated theoretically because there are no anomalous samples available during training.
Therefore, the resulting \ac{asd} performances obtained by minimizing both types of loss functions, angular margin losses and one-class losses, using individual auxiliary classification tasks will be evaluated first.
Furthermore, a combined loss consisting of the sum of the mean of the \ac{ic} compactness losses and an additional softmax layer with a \ac{cce} loss for classification, as proposed in \cite{perera2019learning}, is evaluated.
The results can be found in Tab. \ref{tab:performances}.
Note that it is also possible to divide the classification task into several different classification tasks as for example one task for the machine type and other ones for all or specific attributes \cite{wilkinghoff2021combining,venkatesh2022improved}.
However, in our experience this does not improve performance unless weights for the losses belonging to different machine types are manually tuned to improve the \ac{asd} performance.
Since this requires access to anomalous samples, tuning these weights is impossible in a truly semi-supervised setting.
\par
It can be seen that for both datasets the \ac{asd} performance improves with the number of classes being used for the auxiliary task.
When using only a single class for all data or for individual machine types and sections, the \ac{auc} is close to $50\%$, which corresponds to randomly guessing whether a sample is anomalous or not.
The most likely reason for this is the factory background noise contained in the recordings, which is highly diverse and contains many sound sources other than the target machine.
A model trained with a one-class loss does not know the difference between the sound events emitted by the machines to be monitored and any other sounds contained in the recordings.
The more complex (in terms of numbers of classes) the chosen auxiliary task is, the more information needs to be captured inside the embeddings for solving this task.
Additionally, the background noise does not contain any helpful information for learning to discriminate between the classes defined by the auxiliary task assuming the noise is not class-specific.
As a result, the model learns to monitor specific frequencies or temporal patterns important for specific machine types with specific settings and thus also learns to ignore the background noise and to isolate sounds emitted by the targeted machines.
Furthermore, it can be observed that using an explicit classification task improves performance on all dataset splits.
Ensuring an angular margin between the classes slightly improves the overall performance, but not significantly, often leading to very similar results.
The most likely reason is that by increasing intra-class similarity implicitly introduces a margin between different classes.
Still, using an angular margin loss does not have any drawbacks over using a compactness and a descriptiveness loss.
As a last observation, the sub-cluster AdaCos loss performs slightly better than the AdaCos loss on the development split of the DCASE2022 dataset while yielding a similar performance on the other dataset splits.
A possible explanation that there are no significant improvements on the DCASE2023 datasets when using an angular margin loss is that the auxiliary classification task is not as difficult as for the DCASE2022 dataset because there is only one section for each machine type.
Slight improvements in performance when using multiple sub-clusters for the AdaCos loss have been observed on the DCASE2020 dataset \cite{koizumi2020description} in \cite{wilkinghoff2021sub}.
Note that the DCASE2020 dataset only contains machine recordings with a single parameter setting for each section and no domain shifts, i.e. consists of a single source domain, and thus the task is very different from the much more difficult task considered here.
In conclusion, an angular margin loss for \ac{asd} in combination with an auxiliary classification task that uses as many meaningful classes as possible is an excellent choice when training an \ac{asd} system based on audio embeddings.
\par
In the previous paragraph, we made the assumption that the noise is not class-specific.
However, if there is a single class with very specific noise that is only present for this particular class or, even worse, if this is the case for all classes, then an auxiliary classification task will very likely not improve the results.
The reason is that the model does not learn to closely monitor the machine sound because also the background noise contains useful information for discriminating between the classes.
Therefore, assuming that the noise is not class-specific is essential and intuitively makes sense for machine condition monitoring as one would expect that at least some machines share the same noise distribution
when running in the same factory or acoustic environment.
Moreover, as shown in Theorem \autoref{thm:ang_comp}, minimizing an angular margin loss using an auxiliary classification task also explicitly increases intra-class similarity.
Hence, even if the noise is class-specific and thus the auxiliary classification task does not aid the \ac{asd} task, the performance is still as least as good as when not using a classification task at all but only minimizing the intra-class compactness losses and there should not be a disadvantage.

\subsection{Minimizing Compactness Loss by Minimizing an Angular Margin Loss}
In Theorem \ref{thm:ang_comp}, it has been shown that minimizing an angular margin loss also minimizes all \ac{ic} compactness losses and maximizes all inter-class compactness losses.
This fact is now verified experimentally by training a model using the sub-cluster AdaCos loss while also monitoring all compactness losses.
The results are depicted in Fig. \ref{fig:monitoring_losses} and Fig. \ref{fig:monitoring_losses_no_mixup}.
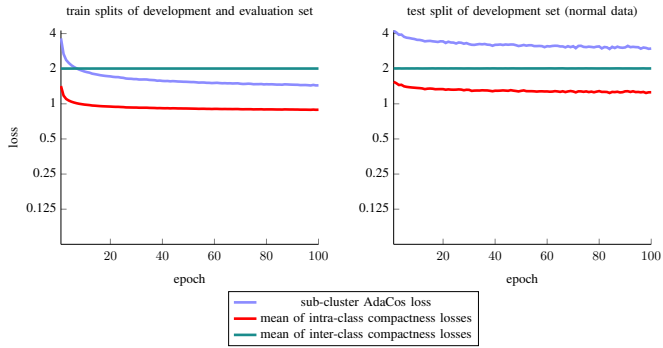
\begin{figure}[t]
    \centering
    \begin{adjustbox}{max width=\columnwidth}
          \begin{tikzpicture}
\begin{groupplot}[
    group style={
    group name=my plots,
    group size=2 by 1,
    xlabels at=edge bottom,
    ylabels at=edge left,
    horizontal sep=2cm,vertical sep=1cm,},
	axis y line*=left,
    axis x line*=bottom,
    xmin=1,
    xmax=100,
    ymode=log,
    ymin=0.0625,
    ymax=4.25,
    legend style={at={(0.65,-0.22)},anchor=north west,legend columns=1},
    ylabel=loss,
    xlabel=epoch,
    xticklabel style={align=center},
    yticklabel style={align=center},
    typeset ticklabels with strut,
    xlabel near ticks,
    ylabel near ticks,
    ytick={0.125,0.25,0.5,1,2,4},
    yticklabels={0.125,0.25,0.5,1,2,4}
]
\nextgroupplot[title=train splits of development and evaluation set]
\legend{sub-cluster AdaCos loss, mean of intra-class compactness losses, mean of inter-class compactness losses}
\addplot[blue!45, mark=none,line width=2pt] coordinates {
(1.0,3.641)(2.0,2.677)(3.0,2.375)(4.0,2.242)(5.0,2.142)(6.0,2.070)(7.0,2.012)(8.0,1.965)(9.0,1.929)(10.0,1.894)(11.0,1.868)(12.0,1.849)(13.0,1.814)(14.0,1.797)(15.0,1.780)(16.0,1.761)(17.0,1.747)(18.0,1.735)(19.0,1.726)(20.0,1.718)(21.0,1.703)(22.0,1.695)(23.0,1.694)(24.0,1.678)(25.0,1.669)(26.0,1.648)(27.0,1.641)(28.0,1.639)(29.0,1.629)(30.0,1.629)(31.0,1.616)(32.0,1.616)(33.0,1.618)(34.0,1.616)(35.0,1.603)(36.0,1.599)(37.0,1.594)(38.0,1.578)(39.0,1.582)(40.0,1.571)(41.0,1.567)(42.0,1.569)(43.0,1.559)(44.0,1.565)(45.0,1.560)(46.0,1.553)(47.0,1.551)(48.0,1.546)(49.0,1.545)(50.0,1.541)(51.0,1.536)(52.0,1.536)(53.0,1.529)(54.0,1.520)(55.0,1.523)(56.0,1.522)(57.0,1.514)(58.0,1.513)(59.0,1.506)(60.0,1.510)(61.0,1.514)(62.0,1.506)(63.0,1.500)(64.0,1.503)(65.0,1.498)(66.0,1.500)(67.0,1.494)(68.0,1.497)(69.0,1.499)(70.0,1.486)(71.0,1.479)(72.0,1.482)(73.0,1.488)(74.0,1.484)(75.0,1.472)(76.0,1.474)(77.0,1.474)(78.0,1.478)(79.0,1.467)(80.0,1.472)(81.0,1.469)(82.0,1.471)(83.0,1.466)(84.0,1.461)(85.0,1.462)(86.0,1.458)(87.0,1.469)(88.0,1.466)(89.0,1.462)(90.0,1.459)(91.0,1.457)(92.0,1.456)(93.0,1.447)(94.0,1.448)(95.0,1.449)(96.0,1.444)(97.0,1.437)(98.0,1.452)(99.0,1.438)(100.0,1.440)};
\addplot[red!100,mark=none,line width=2pt] coordinates {
(1.0,1.421)(2.0,1.177)(3.0,1.101)(4.0,1.067)(5.0,1.043)(6.0,1.026)(7.0,1.012)(8.0,1.001)(9.0,0.993)(10.0,0.985)(11.0,0.980)(12.0,0.975)(13.0,0.967)(14.0,0.964)(15.0,0.960)(16.0,0.956)(17.0,0.953)(18.0,0.951)(19.0,0.949)(20.0,0.947)(21.0,0.944)(22.0,0.942)(23.0,0.942)(24.0,0.939)(25.0,0.937)(26.0,0.932)(27.0,0.931)(28.0,0.930)(29.0,0.929)(30.0,0.929)(31.0,0.926)(32.0,0.926)(33.0,0.926)(34.0,0.926)(35.0,0.923)(36.0,0.922)(37.0,0.921)(38.0,0.918)(39.0,0.919)(40.0,0.917)(41.0,0.916)(42.0,0.917)(43.0,0.914)(44.0,0.916)(45.0,0.915)(46.0,0.913)(47.0,0.913)(48.0,0.912)(49.0,0.912)(50.0,0.911)(51.0,0.910)(52.0,0.909)(53.0,0.908)(54.0,0.906)(55.0,0.907)(56.0,0.907)(57.0,0.905)(58.0,0.905)(59.0,0.903)(60.0,0.905)(61.0,0.905)(62.0,0.903)(63.0,0.902)(64.0,0.903)(65.0,0.902)(66.0,0.902)(67.0,0.901)(68.0,0.902)(69.0,0.902)(70.0,0.899)(71.0,0.898)(72.0,0.899)(73.0,0.899)(74.0,0.899)(75.0,0.896)(76.0,0.897)(77.0,0.897)(78.0,0.898)(79.0,0.895)(80.0,0.896)(81.0,0.896)(82.0,0.896)(83.0,0.895)(84.0,0.894)(85.0,0.894)(86.0,0.893)(87.0,0.896)(88.0,0.895)(89.0,0.894)(90.0,0.894)(91.0,0.893)(92.0,0.893)(93.0,0.891)(94.0,0.891)(95.0,0.891)(96.0,0.890)(97.0,0.889)(98.0,0.892)(99.0,0.889)(100.0,0.890)};
\addplot[teal!90,mark=none,line width=2pt] coordinates {
(1.0,2.004)(2.0,2.005)(3.0,2.005)(4.0,2.005)(5.0,2.005)(6.0,2.005)(7.0,2.006)(8.0,2.006)(9.0,2.006)(10.0,2.006)(11.0,2.006)(12.0,2.005)(13.0,2.006)(14.0,2.006)(15.0,2.006)(16.0,2.006)(17.0,2.006)(18.0,2.006)(19.0,2.006)(20.0,2.006)(21.0,2.006)(22.0,2.006)(23.0,2.006)(24.0,2.006)(25.0,2.006)(26.0,2.006)(27.0,2.006)(28.0,2.006)(29.0,2.006)(30.0,2.006)(31.0,2.006)(32.0,2.006)(33.0,2.006)(34.0,2.006)(35.0,2.006)(36.0,2.006)(37.0,2.006)(38.0,2.006)(39.0,2.006)(40.0,2.006)(41.0,2.006)(42.0,2.006)(43.0,2.006)(44.0,2.006)(45.0,2.006)(46.0,2.006)(47.0,2.006)(48.0,2.006)(49.0,2.006)(50.0,2.006)(51.0,2.006)(52.0,2.006)(53.0,2.006)(54.0,2.006)(55.0,2.006)(56.0,2.006)(57.0,2.006)(58.0,2.006)(59.0,2.006)(60.0,2.006)(61.0,2.006)(62.0,2.006)(63.0,2.006)(64.0,2.006)(65.0,2.006)(66.0,2.006)(67.0,2.006)(68.0,2.006)(69.0,2.006)(70.0,2.006)(71.0,2.006)(72.0,2.006)(73.0,2.006)(74.0,2.006)(75.0,2.006)(76.0,2.006)(77.0,2.006)(78.0,2.006)(79.0,2.006)(80.0,2.006)(81.0,2.006)(82.0,2.006)(83.0,2.006)(84.0,2.006)(85.0,2.006)(86.0,2.006)(87.0,2.006)(88.0,2.006)(89.0,2.006)(90.0,2.006)(91.0,2.006)(92.0,2.006)(93.0,2.006)(94.0,2.006)(95.0,2.006)(96.0,2.006)(97.0,2.006)(98.0,2.006)(99.0,2.006)(100.0,2.006)};
\nextgroupplot[title=test split of development set (normal data)]
\addplot[blue!45,mark=none,line width=2pt] coordinates {
(1.0,4.222)(2.0,4.092)(3.0,3.907)(4.0,3.921)(5.0,3.752)(6.0,3.687)(7.0,3.663)(8.0,3.623)(9.0,3.592)(10.0,3.535)(11.0,3.528)(12.0,3.491)(13.0,3.432)(14.0,3.453)(15.0,3.463)(16.0,3.433)(17.0,3.414)(18.0,3.390)(19.0,3.434)(20.0,3.426)(21.0,3.323)(22.0,3.405)(23.0,3.355)(24.0,3.348)(25.0,3.289)(26.0,3.354)(27.0,3.336)(28.0,3.275)(29.0,3.314)(30.0,3.247)(31.0,3.207)(32.0,3.200)(33.0,3.229)(34.0,3.254)(35.0,3.236)(36.0,3.256)(37.0,3.215)(38.0,3.205)(39.0,3.145)(40.0,3.204)(41.0,3.209)(42.0,3.225)(43.0,3.194)(44.0,3.226)(45.0,3.209)(46.0,3.225)(47.0,3.223)(48.0,3.181)(49.0,3.129)(50.0,3.155)(51.0,3.194)(52.0,3.152)(53.0,3.177)(54.0,3.177)(55.0,3.126)(56.0,3.125)(57.0,3.133)(58.0,3.122)(59.0,3.060)(60.0,3.118)(61.0,3.104)(62.0,3.133)(63.0,3.146)(64.0,3.141)(65.0,3.103)(66.0,3.103)(67.0,3.047)(68.0,3.118)(69.0,3.061)(70.0,3.117)(71.0,2.999)(72.0,3.134)(73.0,3.143)(74.0,3.077)(75.0,3.062)(76.0,3.090)(77.0,3.073)(78.0,3.050)(79.0,3.089)(80.0,3.104)(81.0,2.985)(82.0,3.067)(83.0,3.047)(84.0,2.951)(85.0,3.079)(86.0,3.007)(87.0,3.088)(88.0,3.038)(89.0,3.023)(90.0,3.043)(91.0,3.137)(92.0,3.059)(93.0,3.079)(94.0,3.062)(95.0,3.056)(96.0,2.970)(97.0,3.037)(98.0,3.013)(99.0,2.961)(100.0,2.977)};
\addplot[red!100,mark=none,line width=2pt] coordinates {
(1.0,1.543)(2.0,1.503)(3.0,1.450)(4.0,1.455)(5.0,1.412)(6.0,1.397)(7.0,1.388)(8.0,1.380)(9.0,1.374)(10.0,1.368)(11.0,1.363)(12.0,1.355)(13.0,1.336)(14.0,1.348)(15.0,1.352)(16.0,1.347)(17.0,1.334)(18.0,1.336)(19.0,1.337)(20.0,1.336)(21.0,1.321)(22.0,1.330)(23.0,1.323)(24.0,1.321)(25.0,1.323)(26.0,1.329)(27.0,1.324)(28.0,1.307)(29.0,1.320)(30.0,1.317)(31.0,1.291)(32.0,1.298)(33.0,1.300)(34.0,1.299)(35.0,1.306)(36.0,1.303)(37.0,1.302)(38.0,1.293)(39.0,1.282)(40.0,1.295)(41.0,1.294)(42.0,1.297)(43.0,1.303)(44.0,1.298)(45.0,1.293)(46.0,1.309)(47.0,1.297)(48.0,1.288)(49.0,1.277)(50.0,1.292)(51.0,1.303)(52.0,1.284)(53.0,1.283)(54.0,1.283)(55.0,1.280)(56.0,1.273)(57.0,1.279)(58.0,1.280)(59.0,1.276)(60.0,1.277)(61.0,1.270)(62.0,1.273)(63.0,1.284)(64.0,1.277)(65.0,1.266)(66.0,1.281)(67.0,1.261)(68.0,1.277)(69.0,1.287)(70.0,1.276)(71.0,1.254)(72.0,1.279)(73.0,1.278)(74.0,1.267)(75.0,1.278)(76.0,1.261)(77.0,1.272)(78.0,1.256)(79.0,1.272)(80.0,1.285)(81.0,1.265)(82.0,1.266)(83.0,1.263)(84.0,1.239)(85.0,1.265)(86.0,1.260)(87.0,1.269)(88.0,1.252)(89.0,1.258)(90.0,1.280)(91.0,1.282)(92.0,1.264)(93.0,1.281)(94.0,1.261)(95.0,1.264)(96.0,1.255)(97.0,1.264)(98.0,1.239)(99.0,1.255)(100.0,1.256)};
\addplot[teal!90,mark=none,line width=2pt] coordinates {
(1.0,2.011)(2.0,2.010)(3.0,2.012)(4.0,2.008)(5.0,2.009)(6.0,2.010)(7.0,2.009)(8.0,2.010)(9.0,2.009)(10.0,2.011)(11.0,2.009)(12.0,2.008)(13.0,2.009)(14.0,2.009)(15.0,2.010)(16.0,2.012)(17.0,2.009)(18.0,2.011)(19.0,2.010)(20.0,2.010)(21.0,2.010)(22.0,2.009)(23.0,2.011)(24.0,2.010)(25.0,2.011)(26.0,2.010)(27.0,2.010)(28.0,2.008)(29.0,2.010)(30.0,2.010)(31.0,2.009)(32.0,2.008)(33.0,2.009)(34.0,2.008)(35.0,2.009)(36.0,2.010)(37.0,2.011)(38.0,2.009)(39.0,2.012)(40.0,2.011)(41.0,2.009)(42.0,2.010)(43.0,2.009)(44.0,2.011)(45.0,2.009)(46.0,2.011)(47.0,2.012)(48.0,2.010)(49.0,2.010)(50.0,2.009)(51.0,2.010)(52.0,2.010)(53.0,2.010)(54.0,2.010)(55.0,2.010)(56.0,2.010)(57.0,2.010)(58.0,2.010)(59.0,2.010)(60.0,2.011)(61.0,2.012)(62.0,2.011)(63.0,2.009)(64.0,2.011)(65.0,2.012)(66.0,2.013)(67.0,2.009)(68.0,2.008)(69.0,2.010)(70.0,2.010)(71.0,2.009)(72.0,2.011)(73.0,2.010)(74.0,2.010)(75.0,2.012)(76.0,2.010)(77.0,2.011)(78.0,2.010)(79.0,2.009)(80.0,2.009)(81.0,2.010)(82.0,2.010)(83.0,2.010)(84.0,2.009)(85.0,2.010)(86.0,2.011)(87.0,2.012)(88.0,2.011)(89.0,2.009)(90.0,2.010)(91.0,2.010)(92.0,2.010)(93.0,2.009)(94.0,2.009)(95.0,2.009)(96.0,2.010)(97.0,2.012)(98.0,2.009)(99.0,2.009)(100.0,2.010)};

\end{groupplot}
\end{tikzpicture}
    \end{adjustbox}
    \caption{Different losses after each epoch when training by minimizing sub-cluster AdaCos with a single sub-cluster per class and using mixup.}
    \label{fig:monitoring_losses}
\end{figure}
\begin{figure}[t]
    \centering
    \begin{adjustbox}{max width=\columnwidth}
          \begin{tikzpicture}
\begin{groupplot}[
    group style={
    group name=my plots,
    group size=2 by 1,
    xlabels at=edge bottom,
    ylabels at=edge left,
    horizontal sep=2cm,vertical sep=1cm,},
	axis y line*=left,
    axis x line*=bottom,
    xmin=1,
    xmax=100,
    ymode=log,
    ymin=0.0625,
    ymax=4.25,
    legend style={at={(0.65,-0.22)},anchor=north west,legend columns=1},
    ylabel=loss,
    xlabel=epoch,
    xticklabel style={align=center},
    yticklabel style={align=center},
    typeset ticklabels with strut,
    xlabel near ticks,
    ylabel near ticks,
    ytick={0.125,0.25,0.5,1,2,4},
    yticklabels={0.125,0.25,0.5,1,2,4}
]
\nextgroupplot[title=train splits of development and evaluation set]
\legend{AdaCos loss, mean of intra-class compactness losses, mean of inter-class compactness losses}
\addplot[blue!45, mark=none,line width=2pt] coordinates {
(1.0,2.312)(2.0,1.159)(3.0,1.027)(4.0,1.007)(5.0,1.005)(6.0,1.004)(7.0,0.992)(8.0,0.979)(9.0,0.968)(10.0,0.956)(11.0,0.945)(12.0,0.935)(13.0,0.923)(14.0,0.916)(15.0,0.907)(16.0,0.897)(17.0,0.887)(18.0,0.876)(19.0,0.868)(20.0,0.861)(21.0,0.852)(22.0,0.842)(23.0,0.837)(24.0,0.832)(25.0,0.823)(26.0,0.822)(27.0,0.815)(28.0,0.812)(29.0,0.806)(30.0,0.798)(31.0,0.799)(32.0,0.795)(33.0,0.798)(34.0,0.787)(35.0,0.791)(36.0,0.787)(37.0,0.786)(38.0,0.787)(39.0,0.784)(40.0,0.777)(41.0,0.782)(42.0,0.776)(43.0,0.784)(44.0,0.775)(45.0,0.771)(46.0,0.777)(47.0,0.773)(48.0,0.769)(49.0,0.770)(50.0,0.765)(51.0,0.770)(52.0,0.773)(53.0,0.767)(54.0,0.763)(55.0,0.767)(56.0,0.769)(57.0,0.761)(58.0,0.764)(59.0,0.765)(60.0,0.766)(61.0,0.767)(62.0,0.758)(63.0,0.763)(64.0,0.758)(65.0,0.763)(66.0,0.765)(67.0,0.763)(68.0,0.763)(69.0,0.752)(70.0,0.755)(71.0,0.764)(72.0,0.756)(73.0,0.760)(74.0,0.763)(75.0,0.754)(76.0,0.757)(77.0,0.761)(78.0,0.754)(79.0,0.761)(80.0,0.757)(81.0,0.753)(82.0,0.756)(83.0,0.762)(84.0,0.753)(85.0,0.752)(86.0,0.760)(87.0,0.755)(88.0,0.752)(89.0,0.757)(90.0,0.751)(91.0,0.749)(92.0,0.766)(93.0,0.754)(94.0,0.751)(95.0,0.748)(96.0,0.756)(97.0,0.763)(98.0,0.745)(99.0,0.747)(100.0,0.757)};
\addplot[red!100,mark=none,line width=2pt] coordinates {
(1.0,1.055)(2.0,0.682)(3.0,0.547)(4.0,0.455)(5.0,0.393)(6.0,0.355)(7.0,0.316)(8.0,0.294)(9.0,0.276)(10.0,0.259)(11.0,0.245)(12.0,0.233)(13.0,0.225)(14.0,0.217)(15.0,0.208)(16.0,0.201)(17.0,0.194)(18.0,0.185)(19.0,0.181)(20.0,0.175)(21.0,0.170)(22.0,0.162)(23.0,0.159)(24.0,0.155)(25.0,0.149)(26.0,0.148)(27.0,0.143)(28.0,0.140)(29.0,0.136)(30.0,0.130)(31.0,0.130)(32.0,0.128)(33.0,0.129)(34.0,0.120)(35.0,0.122)(36.0,0.118)(37.0,0.118)(38.0,0.118)(39.0,0.115)(40.0,0.109)(41.0,0.113)(42.0,0.107)(43.0,0.113)(44.0,0.107)(45.0,0.104)(46.0,0.107)(47.0,0.105)(48.0,0.101)(49.0,0.102)(50.0,0.098)(51.0,0.101)(52.0,0.102)(53.0,0.097)(54.0,0.095)(55.0,0.097)(56.0,0.100)(57.0,0.092)(58.0,0.095)(59.0,0.095)(60.0,0.095)(61.0,0.095)(62.0,0.088)(63.0,0.092)(64.0,0.088)(65.0,0.090)(66.0,0.092)(67.0,0.091)(68.0,0.091)(69.0,0.082)(70.0,0.085)(71.0,0.091)(72.0,0.085)(73.0,0.087)(74.0,0.089)(75.0,0.081)(76.0,0.083)(77.0,0.087)(78.0,0.081)(79.0,0.087)(80.0,0.084)(81.0,0.080)(82.0,0.082)(83.0,0.087)(84.0,0.080)(85.0,0.078)(86.0,0.085)(87.0,0.080)(88.0,0.078)(89.0,0.082)(90.0,0.077)(91.0,0.076)(92.0,0.089)(93.0,0.079)(94.0,0.076)(95.0,0.074)(96.0,0.079)(97.0,0.086)(98.0,0.070)(99.0,0.072)(100.0,0.082)};
\addplot[teal!85,mark=none,line width=2pt] coordinates {
(1.0,2.008)(2.0,2.009)(3.0,2.009)(4.0,2.009)(5.0,2.009)(6.0,2.010)(7.0,2.010)(8.0,2.010)(9.0,2.010)(10.0,2.010)(11.0,2.010)(12.0,2.010)(13.0,2.010)(14.0,2.010)(15.0,2.010)(16.0,2.010)(17.0,2.010)(18.0,2.010)(19.0,2.010)(20.0,2.010)(21.0,2.010)(22.0,2.010)(23.0,2.010)(24.0,2.010)(25.0,2.010)(26.0,2.010)(27.0,2.010)(28.0,2.010)(29.0,2.010)(30.0,2.010)(31.0,2.010)(32.0,2.010)(33.0,2.010)(34.0,2.010)(35.0,2.010)(36.0,2.010)(37.0,2.010)(38.0,2.010)(39.0,2.010)(40.0,2.010)(41.0,2.010)(42.0,2.010)(43.0,2.010)(44.0,2.010)(45.0,2.010)(46.0,2.010)(47.0,2.010)(48.0,2.010)(49.0,2.010)(50.0,2.010)(51.0,2.010)(52.0,2.010)(53.0,2.010)(54.0,2.010)(55.0,2.010)(56.0,2.010)(57.0,2.010)(58.0,2.010)(59.0,2.010)(60.0,2.010)(61.0,2.010)(62.0,2.010)(63.0,2.010)(64.0,2.010)(65.0,2.010)(66.0,2.010)(67.0,2.010)(68.0,2.010)(69.0,2.010)(70.0,2.010)(71.0,2.010)(72.0,2.010)(73.0,2.010)(74.0,2.010)(75.0,2.010)(76.0,2.010)(77.0,2.010)(78.0,2.010)(79.0,2.010)(80.0,2.010)(81.0,2.010)(82.0,2.010)(83.0,2.010)(84.0,2.010)(85.0,2.010)(86.0,2.010)(87.0,2.010)(88.0,2.010)(89.0,2.010)(90.0,2.010)(91.0,2.010)(92.0,2.010)(93.0,2.010)(94.0,2.010)(95.0,2.010)(96.0,2.010)(97.0,2.010)(98.0,2.010)(99.0,2.010)(100.0,2.010)};
\nextgroupplot[title=test split of development set (normal data)]
\addplot[blue!45,mark=none,line width=2pt] coordinates {
(1.0,4.039)(2.0,3.783)(3.0,3.707)(4.0,3.633)(5.0,3.550)(6.0,3.492)(7.0,3.446)(8.0,3.382)(9.0,3.378)(10.0,3.365)(11.0,3.367)(12.0,3.373)(13.0,3.379)(14.0,3.386)(15.0,3.359)(16.0,3.315)(17.0,3.341)(18.0,3.151)(19.0,3.373)(20.0,3.244)(21.0,3.304)(22.0,3.307)(23.0,3.310)(24.0,3.286)(25.0,3.287)(26.0,3.089)(27.0,3.183)(28.0,3.295)(29.0,3.254)(30.0,3.235)(31.0,3.233)(32.0,3.229)(33.0,3.248)(34.0,3.261)(35.0,3.262)(36.0,3.196)(37.0,3.143)(38.0,3.216)(39.0,3.237)(40.0,3.154)(41.0,3.214)(42.0,3.231)(43.0,3.168)(44.0,3.115)(45.0,3.176)(46.0,3.203)(47.0,3.205)(48.0,3.139)(49.0,3.168)(50.0,3.179)(51.0,3.185)(52.0,3.216)(53.0,3.123)(54.0,3.197)(55.0,3.161)(56.0,3.158)(57.0,3.189)(58.0,3.181)(59.0,3.155)(60.0,3.090)(61.0,3.230)(62.0,3.206)(63.0,3.124)(64.0,3.100)(65.0,3.199)(66.0,3.193)(67.0,3.200)(68.0,3.118)(69.0,3.122)(70.0,3.071)(71.0,3.177)(72.0,3.214)(73.0,3.208)(74.0,3.282)(75.0,3.228)(76.0,3.281)(77.0,3.277)(78.0,3.244)(79.0,3.301)(80.0,3.129)(81.0,3.236)(82.0,3.052)(83.0,3.171)(84.0,3.209)(85.0,3.213)(86.0,3.178)(87.0,3.257)(88.0,3.244)(89.0,3.279)(90.0,3.075)(91.0,3.221)(92.0,3.251)(93.0,3.237)(94.0,3.066)(95.0,3.258)(96.0,3.272)(97.0,3.267)(98.0,3.241)(99.0,2.969)(100.0,3.266)};
\addplot[red!100,mark=none,line width=2pt] coordinates {
(1.0,1.461)(2.0,1.357)(3.0,1.321)(4.0,1.273)(5.0,1.259)(6.0,1.222)(7.0,1.210)(8.0,1.197)(9.0,1.179)(10.0,1.170)(11.0,1.160)(12.0,1.147)(13.0,1.152)(14.0,1.147)(15.0,1.134)(16.0,1.135)(17.0,1.124)(18.0,1.115)(19.0,1.126)(20.0,1.111)(21.0,1.116)(22.0,1.112)(23.0,1.103)(24.0,1.100)(25.0,1.099)(26.0,1.086)(27.0,1.095)(28.0,1.111)(29.0,1.087)(30.0,1.081)(31.0,1.087)(32.0,1.100)(33.0,1.085)(34.0,1.087)(35.0,1.089)(36.0,1.064)(37.0,1.077)(38.0,1.081)(39.0,1.075)(40.0,1.068)(41.0,1.080)(42.0,1.088)(43.0,1.079)(44.0,1.053)(45.0,1.064)(46.0,1.093)(47.0,1.067)(48.0,1.057)(49.0,1.063)(50.0,1.077)(51.0,1.062)(52.0,1.065)(53.0,1.063)(54.0,1.077)(55.0,1.044)(56.0,1.069)(57.0,1.064)(58.0,1.046)(59.0,1.064)(60.0,1.034)(61.0,1.068)(62.0,1.067)(63.0,1.063)(64.0,1.044)(65.0,1.060)(66.0,1.067)(67.0,1.067)(68.0,1.040)(69.0,1.069)(70.0,1.071)(71.0,1.065)(72.0,1.059)(73.0,1.069)(74.0,1.061)(75.0,1.053)(76.0,1.071)(77.0,1.060)(78.0,1.062)(79.0,1.095)(80.0,1.064)(81.0,1.060)(82.0,1.065)(83.0,1.055)(84.0,1.046)(85.0,1.060)(86.0,1.074)(87.0,1.049)(88.0,1.069)(89.0,1.080)(90.0,1.059)(91.0,1.060)(92.0,1.075)(93.0,1.063)(94.0,1.053)(95.0,1.085)(96.0,1.084)(97.0,1.065)(98.0,1.061)(99.0,1.061)(100.0,1.083)};
\addplot[teal!90,mark=none,line width=2pt] coordinates {
(1.0,2.007)(2.0,2.005)(3.0,2.008)(4.0,2.007)(5.0,2.006)(6.0,2.007)(7.0,2.008)(8.0,2.006)(9.0,2.006)(10.0,2.007)(11.0,2.008)(12.0,2.007)(13.0,2.007)(14.0,2.008)(15.0,2.008)(16.0,2.007)(17.0,2.007)(18.0,2.008)(19.0,2.007)(20.0,2.008)(21.0,2.008)(22.0,2.007)(23.0,2.007)(24.0,2.008)(25.0,2.008)(26.0,2.008)(27.0,2.008)(28.0,2.008)(29.0,2.007)(30.0,2.008)(31.0,2.008)(32.0,2.007)(33.0,2.008)(34.0,2.008)(35.0,2.007)(36.0,2.007)(37.0,2.007)(38.0,2.008)(39.0,2.007)(40.0,2.008)(41.0,2.008)(42.0,2.007)(43.0,2.008)(44.0,2.008)(45.0,2.008)(46.0,2.007)(47.0,2.007)(48.0,2.007)(49.0,2.008)(50.0,2.007)(51.0,2.007)(52.0,2.007)(53.0,2.008)(54.0,2.008)(55.0,2.008)(56.0,2.007)(57.0,2.008)(58.0,2.009)(59.0,2.007)(60.0,2.008)(61.0,2.008)(62.0,2.008)(63.0,2.007)(64.0,2.008)(65.0,2.008)(66.0,2.009)(67.0,2.008)(68.0,2.008)(69.0,2.008)(70.0,2.008)(71.0,2.008)(72.0,2.008)(73.0,2.009)(74.0,2.008)(75.0,2.009)(76.0,2.008)(77.0,2.008)(78.0,2.008)(79.0,2.008)(80.0,2.009)(81.0,2.008)(82.0,2.008)(83.0,2.008)(84.0,2.009)(85.0,2.008)(86.0,2.009)(87.0,2.008)(88.0,2.008)(89.0,2.008)(90.0,2.007)(91.0,2.008)(92.0,2.008)(93.0,2.008)(94.0,2.009)(95.0,2.008)(96.0,2.008)(97.0,2.008)(98.0,2.008)(99.0,2.008)(100.0,2.008)};

\end{groupplot}
\end{tikzpicture}
    \end{adjustbox}
    \caption{Different losses after each epoch when training by minimizing AdaCos and not using mixup.}
    \label{fig:monitoring_losses_no_mixup}
\end{figure}
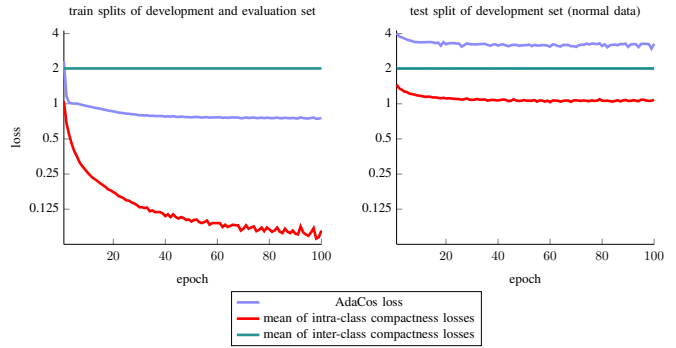
Regardless of the dataset splits and regardless of using or not using mixup, the angular margin loss and the mean of the \ac{ic} compactness losses are decreasing during training.
The mean of the inter-class compactness loss is constantly equal to $2$, even without training.
The reason is that all sub-cluster centers in this work are constant, randomly initialized and projected to the unit sphere.
Hence, By Lemma \ref{lem:cos_id}, a squared Euclidean distance of $2$ corresponds to an angle of $\frac{\pi}{2}$, i.e. orthogonality.
The most likely reason is that the randomly initialized center vectors are approximately orthogonal with very high probability because of the high dimension $D=256$ of the embedding space.
Thus, samples that are similar to the center of one class will be approximately orthogonal to the centers of the other classes.
Overall, this is exactly the expected behavior as predicted by Theorem \ref{thm:ang_comp} and therefore verifies the theoretical results.
Note that smaller loss values do not correspond to a better \ac{asd} performance because minimizing these losses only optimizes the performance for the auxiliary task, which is not the same as the \ac{asd} task.

\subsection{Visualizing Normal and Anomalous Regions in Input Representations as Perceived by the System}
\begin{figure*}[t]
\centering
\begin{adjustbox}{max width=\textwidth}
\begin{tabular}{ccc}
\subfloat[Spectrogram of an anomalous gearbox sound.]{\includegraphics{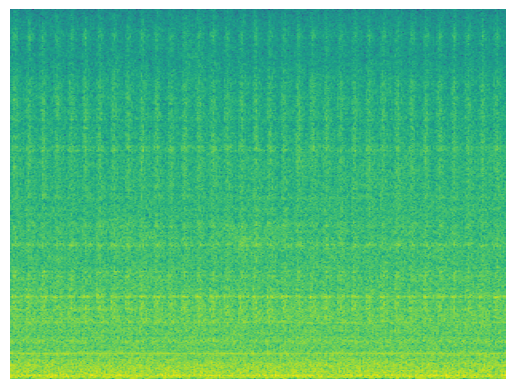}\label{fig:rise_a}} & 
\subfloat[Importance map of an anomalous gearbox sound when using sub-cluster AdaCos.]{\includegraphics{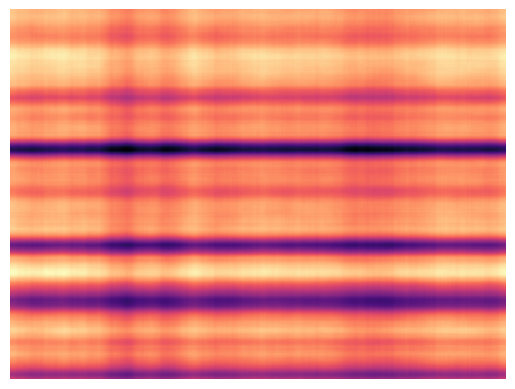}\label{fig:rise_g}} & 
\subfloat[Importance map of an anomalous gearbox sound when using compactness loss.]{\includegraphics{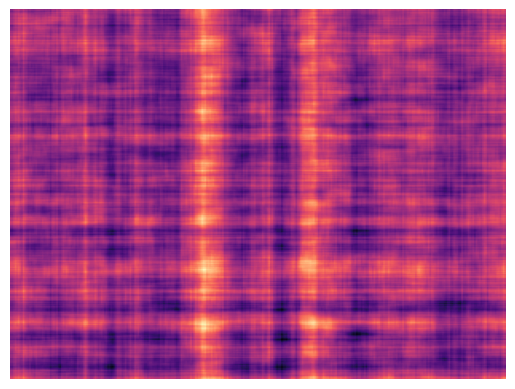}\label{fig:rise_d}}\\
\subfloat[Spectrogram of a normal valve sound.]{\includegraphics{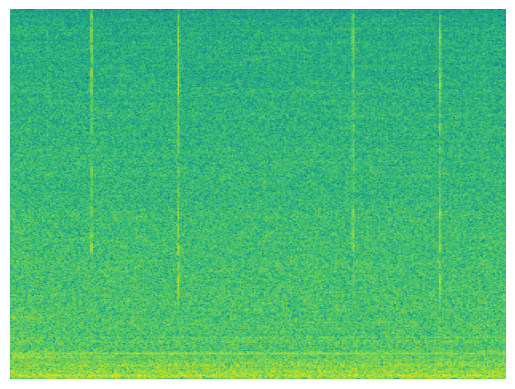}\label{fig:rise_b}}&
\subfloat[Importance map of a normal valve sound when using sub-cluster AdaCos.]{\includegraphics{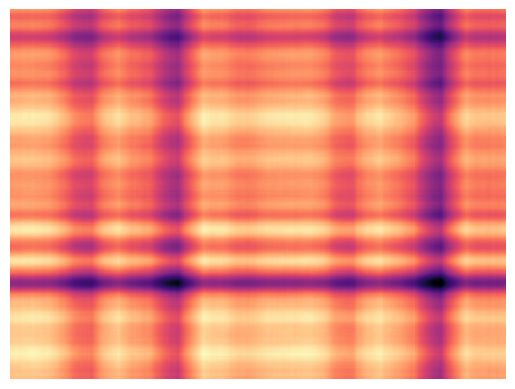}\label{fig:rise_h}}&
\subfloat[Importance map of a normal valve sound when using compactness loss.]{\includegraphics{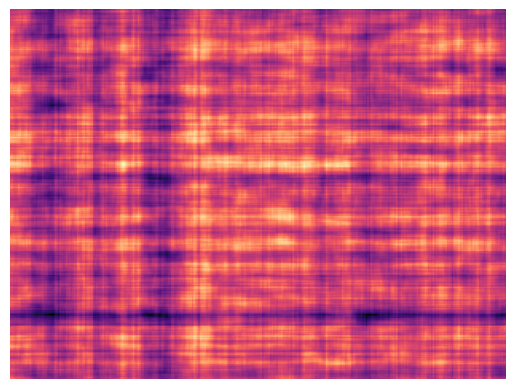}\label{fig:rise_e}}\\
\end{tabular}
\end{adjustbox}
\caption{Log scaled spectrograms (left column), importance maps obtained with \ac{rise} when training with the sub-cluster AdaCos loss and classifying between different machine types, sections and attribute information (middle column), and importance maps obtained with \ac{rise} when training with an \ac{ic} compactness loss and no auxiliary classification task (right column) for two different recordings belonging to the test split of the development set (rows). For the importance maps, blue colors indicate normal regions and yellow colors indicate regions that are found to be anomalous by the model. All subfigures use individual color scales to improve visual appearance for differently scaled importance maps and thus colors of different subfigures cannot be compared to each other.}
    \label{fig:rise_plots}
\end{figure*}
\begin{figure*}[t]
	\centering
    \begin{adjustbox}{max width=\textwidth}
    \includegraphics[page=1]{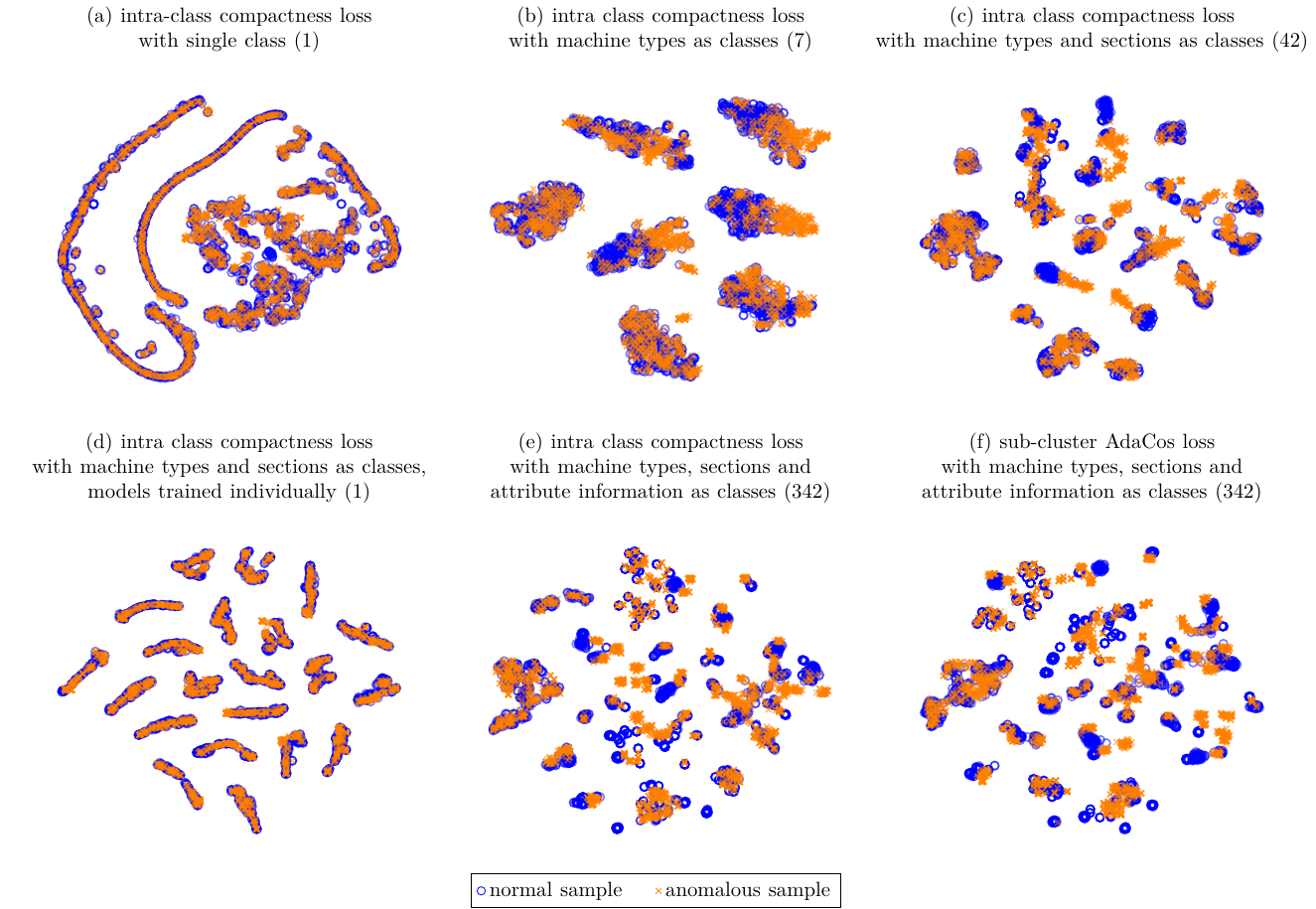}
    \end{adjustbox}
    \caption{Visualizations of the test split of the development set in the learned embedding space for different loss functions and auxiliary tasks using \ac{tsne}. Numbers in brackets denote the number of different classes used for the auxiliary task.}
    \label{fig:tsne}
\end{figure*}
To further investigate the effect of using an auxiliary task with multiple classes, another experiment using \ac{rise} \cite{petsiuk2018rise} is carried out.
\Ac{rise} highlights regions of the input representations that are considered normal or anomalous by the \ac{asd} system.
Our goal is to show that utilizing an auxiliary classification task for training the system, as done when minimizing an angular margin loss, enables the system to closely monitor specific machine sounds by focusing on regions belonging to specific patterns of the input data.
Although the \ac{asd} performance is worse when only using spectrograms as input representations \cite{wilkinghoff2023design}, for these experiments a model using only spectrograms as input has been trained.
The reason is that these representations are visually more appealing for the human eye than waveforms or spectra and thus more suitable to visually highlight normal and anomalous regions.
\par
To visualize areas of the input representation responsible for a decision, \ac{rise} masks random entries of the spectrograms using binary masks and evaluates the \ac{asd} score using the masked spectrogram.
This step is repeated for many iterations.
Then, the sum of the masks weighted with the corresponding \ac{asd} scores is taken and normalized with the expected value of a random binary mask, which depends on the chosen sampling distribution.
The result is called an \emph{importance map} and visualizes the impact of specific regions of a spectrogram on the resulting anomaly score.
\par
The problem is that the dimension of the spectrograms is very high because a time dimension of $T=311$ and a frequency dimension of $F=513$ is used.
Thus, there are $2^{T\cdot F}=2^{159543}$ possible binary masks and thus \ac{rise} requires clearly too many iterations.
To significantly reduce the search space from $2^{F\cdot T}$ to $2^{F+T}$, individual time and frequency masks are randomly generated with a probability of $0.25$ for a time step or frequency bin to be masked and  both masks are combined by element-wise multiplication.
This restriction is not too severe because most sounds emitted by machines are relatively stable over time with specific frequencies (e.g. fans), consist of multiple stable sound events with on- and offsets (e.g. slide rails) or only consist of short sound events over a wide frequency range with a specific temporal structure (e.g. valves).
For further reduction of the search space, small binary masks are generated and then up-sampled and randomly cropped to match the dimension of the spectrogram to be masked as proposed in \cite{petsiuk2018rise}.
More concretely, we used time masks of size $20$ and frequency masks of size $34$ resulting in a search space of $2^{54}$, which is still very large but much smaller than before.
For generating a single importance map, $640,000$ iterations have been used.
\par
Magnitude spectrograms (visualized in log scale) and corresponding importance maps belonging to two different samples using i) a model trained with an \ac{ic} compactness loss without an auxiliary task, and ii) a model trained with the sub-cluster AdaCos loss and an auxiliary task for classifying between different machine types, sections and attribute information are depicted in Fig. \ref{fig:rise_plots}.
For the depicted importance maps, blue colors indicate normal regions and yellow colors indicate anomalous regions as perceived by the system.
Note that, since the system does not yield perfect results, these regions do not need to really belong to normal and anomalous regions.
As there are only binary labels, indicating normal or anomalous samples, available for each entire audio recording and we are no subject matter experts for machine condition monitoring, we do not know which regions are normal or anomalous.
Still, for the purpose of showing that utilizing meta information when training a model, as done by angular margin losses, helps the system to have a better understanding of the structure of the data these plots are sufficient.
There are several observations to be made.
Comparing the representations depicted in Fig. \ref{fig:rise_g} and \ref{fig:rise_h} with the ones depicted in Fig. \ref{fig:rise_d} and \ref{fig:rise_e}, we suggest that using sub-cluster AdaCos, i.e. Fig. \ref{fig:rise_g} and \ref{fig:rise_h}, more clearly shows time and frequency structures at a resolution correlating with the structures resp. acoustic events visible in the spectrograms depicted in Fig. \ref{fig:rise_a} and \ref{fig:rise_b}.
\par
For the anomalous gearbox example (Fig. \ref{fig:rise_a}), the importance map depicted in Fig. \ref{fig:rise_g} shows that specific frequencies are monitored and considered to be normal or anomalous.
Interestingly, the normal frequency regions (in blue) in Fig. \ref{fig:rise_g} exactly correspond to the frequencies containing high energy (Fig. \ref{fig:rise_a}) showing that the model expects a gearbox sound from this section to have high energy in these regions.
The frequencies that are considered most anomalous, which mostly corresponds to the frequency range between the bottom two normal frequency bands, only contain some energy.
This indicates that a normal machine sound should either contain no energy or much more energy for these frequencies.
In contrast to this, the importance map depicted in Fig. \ref{fig:rise_d} does not monitor specific frequencies and the only clearly visible structures are two vertical lines indicating anomalous regions (in yellow).
Although we cannot guarantee that the regions in the spectrogram corresponding to these vertical lines are not anomalous, at least visually there is no energy present in these locations.
Since the recordings of the machine sounds do not start and end at the same fixed time steps, it does not make sense that the model expects temporal patterns at exactly these time steps that are missing and to thus consider such patterns to be anomalous.
Therefore, it seems that these structures are errors of the model.
\par
The importance maps belonging to the normal valve example (Fig. \ref{fig:rise_b}) show a similar behavior but for temporal patterns in addition to specific frequencies.
Here, the main four normal vertical patterns in the importance map shown in Fig. \ref{fig:rise_h} correspond the four high energy patterns of the spectrogram showing that the system views these temporal patterns as normal for a valve sound.
In contrast, the importance map depicted in Fig. \ref{fig:rise_e} does not show that the system has learned to detect these patterns and looks almost random.
\par
Overall, the depicted results add further confidence to the claim that training a model with an auxiliary classification task with many classes enables the model to learn much more meaningful embeddings, also leading to much better capabilities for detecting anomalous sound events than a model trained with only a single class.

\subsection{Visualizing the Resulting Embedding Spaces Using t-SNE}
\begin{table}[t]
	\centering
	\caption{Mean and standard deviation of the average Euclidean distance between the \ac{tsne} projections of each anomalous sample and the closest normal sample over five trials for different losses and using different auxiliary tasks.}
\begin{adjustbox}{max width=\columnwidth}
	\begin{tabular}{lll}
		\toprule
        loss & classes of auxiliary task (number of classes) & average distance\\
		\midrule
		\ac{ic} compactness loss&none ($1$)&$0.485\pm0.007$\\
		\ac{ic} compactness loss&machine types ($7$)&$1.636\pm0.037$\\
		\ac{ic} compactness loss&machine types and sections ($42$)&$2.175\pm0.075$\\
		\ac{ic} compactness loss&machine types and sections, models trained individually ($1$)&$0.559\pm0.002$\\
		\ac{ic} compactness loss&machine types, sections and attribute information ($342$)&$2.646\pm0.045$\\
		sub-cluster AdaCos loss&machine types, sections and attribute information ($342$)&$2.947\pm0.022$\\
		\bottomrule
	\end{tabular}
\end{adjustbox}
\label{tab:distances}
\end{table}
As a last experiment, the embedding spaces resulting from using different loss functions and auxiliary tasks are visualized in Figure \ref{fig:tsne} using \ac{tsne} \cite{vandermaaten2008visualizing}.
Note that by Lemma \ref{lem:cos_id} it does not matter whether \ac{tsne} is evaluated with the cosine distance or the Euclidean distance because both are equivalent when determining the degree of similarity between samples on the unit sphere.
It can be seen that using more classes for the auxiliary task helps to separate normal and anomalous samples (Fig. \ref{fig:tsne}b,c,e,f).
When only using a single class (Fig. \ref{fig:tsne}a) or individually trained models (Fig. \ref{fig:tsne}d), there is no visual difference between normal and anomalous samples.
However, it can also be seen that the model has not learned a trivial solution as the embedding spaces did not collapse to a single fixed point, which would correspond to a uniformly distributed \ac{tsne} embedding space.
Moreover, the \ac{asd} performance would be very close to $50\%$ as normal and anomalous samples would be indistinguishable in the embedding space.
Therefore, the applied regularization strategies, namely not using trainable centers and not using bias terms, work and a completely failed regularization is not the main underlying problem.
These visual impressions are verified by computing the average Euclidean distance between each anomalous sample and the closest normal sample in the \ac{tsne} embedding space.
The results can be found in Tab. \ref{tab:distances} and also agree with the performance results shown in Table \ref{tab:performances}.
Note that the distance in the original embedding space is implicitly captured by the ASD performance given in \ref{tab:performances} because the anomaly score is computed by taking the distance to the closest normal sample in the target domain and the closest mean in the source domain.
Again, the most likely explanation for the strong differences between the embedding spaces in terms of \ac{asd} capabilities is that using multiple classes enables the model to focus less on or even ignore the background noise and isolate the targeted machine sounds.
This helps the model to more robustly detect deviations from normal machine sounds despite the acoustically noisy recording conditions and thus results in better \ac{asd} performance.

\section{Conclusions}
\label{sec:conclusions}
In this work, it has been investigated why using angular margin losses works well for semi-supervised \ac{asd}.
To this end, it has been shown, both theoretically and experimentally, that reducing an angular margin loss also minimizes the \ac{ic} compactness loss while simultaneously maximizing the inter-class compactness loss.
Therefore, angular margin losses in combination with an auxiliary classification task can be viewed as regularized one-class losses preventing the model to learn trivial solutions.
In experiments conducted on the DCASE2022 and DCASE2023 \ac{asd} datasets for machine condition monitoring, it has been shown that using an auxiliary task with as many meaningful classes as possible and using an angular margin loss leads to significantly better \ac{asd} performance than using a one-class loss such as the \ac{ic} compactness loss.
Furthermore, \ac{rise} has been applied to create importance maps for different losses and \ac{tsne} has been used to visualize the resulting embedding spaces.
All the conducted experiments show that by using an angular margin the model used for extracting the embeddings learns to monitor relevant frequency bins and learns machine-specific temporal patterns. This enables the model to isolate machine sounds and effectively ignore background noise present in the recording explaining why angular margin losses with an auxiliary task are a good choice for training an \ac{asd} system.
\par
For future work, is is planned to investigate whether using auxiliary tasks based on self-supervised learning to obtain suitable representations of the data improves the resulting \ac{asd} performance.
In addition, sophisticated methods for visualizing anomalous regions of input representations should be developed as being able to localize these regions is very useful for practical applications and theoretical analysis of \ac{asd} systems.

\section*{Acknowledgments}
We would like to thank Paul M. Baggenstoss and Lukas Henneke as well as the anonymous reviewers for their valuable comments that improved the quality of this work.

{\appendix
\label{sec:appendix}
\section*{Proof of Lemma \ref{lem:cos_id}}
\noindent Using only basic definitions, we obtain
\BEs\lVert x-y\rVert_2^2&=\sum_{i=1}^D(x_i-y_i)^2\\&=\sum_{i=1}^Dx_i^2+\sum_{i=1}^Dy_i^2-2\sum_{i=1}^Dx_iy_i\\
&=\lVert x\rVert_2^2+\lVert y\rVert_2^2-2\langle x,y\rangle\\
&=2\bigg(1-\frac{\langle x,y\rangle}{\lVert x\rVert_2\lVert y\rVert_2}\bigg)\\
&=2(1-\cos(x,y)),\EEs
which finishes the proof.
\hfill\qed
}

\bibliographystyle{IEEEtran}
\bibliography{refs}


\begin{IEEEbiography}[{\includegraphics[width=1in,height=1.25in,clip,keepaspectratio]{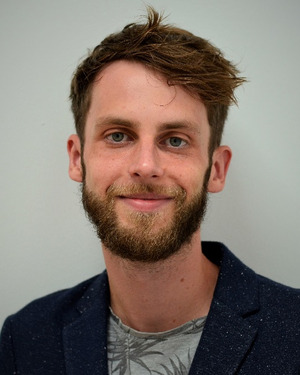}}]{Kevin Wilkinghoff} received his B.Sc. degree in Mathematics at the University of M\"unster, Germany, and his M.Sc. degree in Computer Science at the University of Bonn, Germany, in 2014 and 2017, respectively. 
Since 2017 he is a research associate at Fraunhofer FKIE.
Currently, he is working towards a Ph.D. degree in Computer Science at the University of Bonn.
His research interests include anomaly detection, open-set classification and representation learning for machine listening applications.
In 2021, he received the DCASE Best Paper Award.
\end{IEEEbiography}

\vspace{11pt}

\begin{IEEEbiography}[{\includegraphics[width=1in,height=1.25in,clip,keepaspectratio]{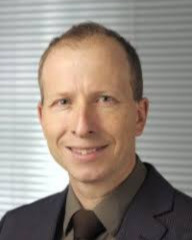}}]{Frank Kurth} studied Computer Science and Mathematics at Bonn University, Germany, where he recieved both a masters degree in Computer Science and the degree of a doctor of natural sciences (Dr. rer. nat.) in 1997 and 1999, respectively.
 From 1997-2007 he was with the Multimedia Signal Processing group at Bonn University where he finished his Habilitation in Computer Science in 2004 and was subsequently appointed apl. Professor in 2013. Since 2007 he is with Fraunhofer FKIE, Germany, where he currently heads a research group focused on physical layer signal analysis in the area of communications.
His research interests include the application of pattern recognition and machine learning techniques to audio, speech and communication signal processing. Dr. Kurth has recieved the 2000 Dissertation Award of the German Informatics Society (GI) and a 2000 Multimedia Award of the German Department of Economy and Technology. He is co-author of more than 100 publications and holds several patents. Dr. Kurth is a senior member of the IEEE.

\end{IEEEbiography}

\vfill

\end{document}